%% file: hitting.tex
\documentclass[11pt]{article}

\input{macros} \usepackage{comment}

\title{Optimal Hitting Sets for Combinatorial Shapes}

\author{Aditya Bhaskara\thanks{Department of Computer Science, Princeton
University. Email: {\sf bhaskara@cs.princeton.edu}} \and Devendra
Desai\thanks{Department of Computer Science, Rutgers University. Email: {\sf
devdesai@cs.rutgers.edu}} \and Srikanth Srinivasan\thanks{Department of Mathematics, Indian Institute of Technology Bombay. Email: {\sf srikanth@math.iitb.ac.in}. This work was done when the author was a postdoctoral researcher at DIMACS, Rutgers University.}}

\begin{document}

\maketitle \input{abstract} %\newpage

\section{Introduction} \input{intro.tex}

\section{Notation and Preliminaries} \input{prelims.tex}

\section{Outline of the Construction} \input{outline.tex}
\section{Hitting sets for Combinatorial Thresholds} As described above, we first
consider the high-weight case (i.e., $w(f) \ge C\log n$ for some large absolute
constant $C$). Next, we consider the low-weight case, with an additional
restriction that each of the accepting probabilities $p_i \le 1/2$. This serves
as a good starting point to explain the {\em general} low-weight case, which we
get to in Section~\ref{sec:lowwt_general}. In each section, we outline our construction and then analyze it for a generic combinatorial threshold $f:[m]^n\rightarrow \{0,1\}$ (subject to weight constraints) defined using sets $A_1,\ldots,A_n\subseteq [m]$. The theorem we finally prove in the section is as follows.

\begin{theorem} \label{thm:threshold-fool} 
For any constant $c\geq 1$, the following holds. 
Suppose $m,1/\eps \leq n^c $. For the class of functions $ \cthr(m,n)$, there
exists an explicit $\eps$-hitting set of size $n^{O_c(1)}$.
\end{theorem}

The main theorem, which we state below, follows directly from the statements of Theorem \ref{thm:threshold-fool} and Lemmas \ref{lem:params} and \ref{lem:cthr_suff}.

\begin{theorem}\label{thm:main-formal}
For any $m,n\in\naturals$ and $\eps > 0$, there is an explicit $\eps$-hitting set for $\cshape(m,n)$ of size $\poly(mn/\eps)$.
\end{theorem}

\input{highwt.tex} \input{lowwt-lowsize.tex} \input{lowwt-general.tex}

\paragraph{Proof of Theorem~\ref{thm:threshold-fool}.} 
The theorem follows easily from Theorems \ref{thm_high_wt} and \ref{thm_low_wt}.
Fix constant $c\geq 1$ s.t. $m,1/\varepsilon\leq n^c$. For $C>0$ a constant
depending on $c$, we obtain hitting sets for thresholds of weight at least
$C\log n$ from Theorem \ref{thm_high_wt} and for thresholds of weight at most
$C\log n$ from Theorem \ref{thm_low_wt}. Their union is an $\varepsilon$-HS for
all of $\cthr(m,n)$.

\section{Stronger Hitting sets for Combinatorial
Rectangles}\label{sec:comb-rectangles} \input{stronger-rect.tex}

\section{Constructing a Fractional Perfect Hash family}\label{sec:fract-hash} \input{fract-hash.tex}

\section{\textbf{Open Problems.}}\label{sec:open-problems} 
We have used a two-level hashing procedure to construct hitting sets for
combinatorial thresholds of low weight. It would be nice to obtain a simpler
construction avoiding the use of an `inner' hitting set construction.

It would also be nice to extend our methods to weighted variants of
combinatorial shapes: functions which accept an input $x$ iff $\sum_i \alpha_i
\one_{A_{i}}(x_i) = S$ where $\alpha_i \in \real_{\ge 0}$. The difficulty here
is that having hitting sets for this sum being $\ge S$ and $\le S$ do not imply
a hitting set for `$=S$'. The simplest open case here is $m=2$ and all $A_i$
being $\{1\}$.\footnote{Note that by the pigeon hole principle, such $S$ exist
for every choice of the $\alpha_i$.} However, it would also be interesting to
prove formally that such weighted versions can capture much stronger
computational classes. 

%\appendix 

%\input{appendix.tex}

\bibliographystyle{plain} 
\bibliography{symmetric}

\appendix
\input{appendix.tex}

\end{document}

%% file: macros.tex
\def\showkeys{0}

\usepackage{fullpage}
\usepackage{xspace,xcolor,enumerate}
\usepackage{amsmath,amsfonts,amssymb,amsthm}
\usepackage{color,graphicx}

\ifnum\showkeys=1
\usepackage[color]{showkeys}
\fi

\usepackage[colorlinks,linkcolor=blue,filecolor=blue,citecolor=blue,urlcolor=blue]{hyperref}
\newtheorem{theorem}{Theorem}[section]

\newtheorem{definition}[theorem]{Definition}
\newtheorem{lemma}[theorem]{Lemma}
\newtheorem{remark}[theorem]{Remark}

\newtheorem{corollary}[theorem]{Corollary}
\newtheorem{claim}[theorem]{Claim}
\newtheorem{fact}[theorem]{Fact}

\def\to{\rightarrow}
\def\eps{\varepsilon}
\def\del{\delta}
\def\epsilon{\varepsilon}

\def\phi{\varphi}
\def\cal{\mathcal}

\renewcommand{\bar}{\overline}

\newcommand{\cshape}{\mathsf{CShape}}
\newcommand{\hs}{\mathsf{HS}}
\newcommand{\crect}{\mathsf{CRect}}
\newcommand{\cthr}{\mathsf{CThr}}

\newcommand{\etal}{et al.\xspace}

\newcommand{\R}{{\mathbb R}}

\newcommand{\E}{{\mathbb E}}

\newcommand{\norm}[1]{\ensuremath{\lVert #1 \rVert}}
\newcommand{\one}{{\mathbf{1}}}
\newcommand{\Psymb}{\mathbb{P}}

\DeclareMathOperator*{\ProbOp}{\Psymb}
\renewcommand{\Pr}{\ProbOp}

\newfont{\inhead}{eufm10 scaled\magstep1}

\newcommand{\calS}{{\cal S}}
\newcommand{\calH}{{\cal H}}
\newcommand{\calR}{{\cal R}}
\newcommand{\calW}{{\cal W}}

\newcommand{\poly}{{\mathrm{poly}}}

\newcommand{\real}{\mathbb{R}}

\newcommand{\naturals}{\mathbb{N}}

\newcommand{\prob}[2]{\ProbOp_{#1}\left[#2\right]}

\newcommand{\setcond}[2]{\left\{#1\: \middle|\: #2\right\}}

\newcommand{\mc}[1]{\mathcal{#1}}

\newcommand{\avg}[2]{\mathop{\textbf{E}}_{#1}[#2]}

\newcommand{\zeroone}{\{0,1\}}

\renewcommand{\L}{\mbox{\rm L}}

\newcommand{\RL}{\mbox{\rm RL}}

\newcommand{\ones}{\mathbf{1}}
\newcommand{\parr}[1]{#1^{\parallel}}
\newcommand{\prp}[1]{#1^{\perp}}

%% file: abstract.tex
\begin{abstract}

We consider the problem of constructing explicit Hitting sets for Combinatorial Shapes, a class of statistical tests first studied by Gopalan, Meka, Reingold, and Zuckerman (STOC 2011). These generalize many well-studied classes of tests, including symmetric functions and combinatorial rectangles. Generalizing results of Linial, Luby, Saks, and Zuckerman (Combinatorica 1997) and Rabani and Shpilka (SICOMP 2010), we construct hitting sets for Combinatorial Shapes of size polynomial in the alphabet, dimension, and the inverse of the error parameter. This is optimal up to polynomial factors. The best previous hitting sets came from the Pseudorandom Generator construction of Gopalan et al., and in particular had size that was quasipolynomial in the inverse of the error parameter.

Our construction builds on natural variants of the constructions of Linial et al. and Rabani and Shpilka. In the process, we construct fractional perfect hash families and hitting sets for combinatorial rectangles with stronger guarantees. These might be of independent interest.
 
\end{abstract}

%% file: intro.tex
Randomness is a tool of great importance in Computer Science and combinatorics. The probabilistic method is highly effective both in the design of simple and efficient 
algorithms and in demonstrating the existence of combinatorial objects with interesting 
properties. But the use of randomness also comes with some 
disadvantages. In the setting of algorithms, introducing randomness adds to the 
number of resource requirements of the algorithm, since truly random bits are hard to come by. For combinatorial constructions, `explicit' versions of these objects often turn out to 
have more structure, which yields advantages beyond the mere fact of their existence (e.g., we know of explicit error-correcting codes that can be efficiently encoded and decoded, but we 
don't know if random codes can \cite{BlumKalaiWasserman}). Thus, it makes sense to ask exactly how powerful probabilistic algorithms and arguments are. 
Can they be `derandomized', i.e., replaced by determinstic algorithms/arguments of comparable 
efficiency?\footnote{A `deterministic argument' for the existence of a combinatorial object is one 
that yields an efficient deterministic algorithm for its construction.} There is a long line of 
research that has addressed this question in various forms  \cite{NW,IW97,Nisan92,SU06,MoserTardos}.

An important line of research into this subject is the question of derandomizing randomized 
space-bounded algorithms. In 1979, Aleliunas  et al.
\cite{AleliunasKaLiLoRa79} demonstrated the power of these algorithms by showing that undirected 
$s$-$t$ connectivity can be solved by randomized algorithms in just $O(\log n)$ space. In order to show that any randomized logspace computation could be derandomized within the same space requirements, researchers considered the problem of 
constructing an efficient \emph{$\varepsilon$-Pseudorandom Generator} ($\varepsilon$-PRG) that 
would stretch a short random seed to a long pseudorandom string that would be indistinguishable 
(up to error $\varepsilon$) to any logspace algorithm.\footnote{As a function of its random bits, the logspace algorithm is \emph{read-once}: it scans its 
input once from left to right.} In particular, an $\varepsilon$-PRG (for small constant 
$\varepsilon>0$) with seedlength $O(\log n)$ would allow efficient deterministic simulations of 
logspace randomized algorithms since a deterministic algorithm could run over all possible random 
seeds.

A breakthrough work of Nisan \cite{Nisan92} took a massive step towards this goal by giving an 
explicit $\varepsilon$-PRG for $\varepsilon=1/\poly(n)$ that stretches $O(\log^2 n)$ truly random 
bits to an $n$-bit pseudorandom string for logspace computations. In the two decades since, 
however, Nisan's result has not been improved upon at this level of generality. However, many
interesting subcases of this class of functions have been considered as avenues for 
progress \cite{NZ,LLSZ,Lu02,LRTV,MekaZu09}.

The class of functions we consider are the very natural class of 
\emph{Combinatorial 
Shapes}. A boolean function $ f $ is a combinatorial shape if it takes $ n $ inputs $ x_1,\ldots,x_n \in 
[m] $ and computes a symmetric function of boolean bits $y_i$ that depend on 
the 
membership of the inputs $ x_i $ in sets $ A_i \subseteq [m] $ associated with $ f $. (A function 
of boolean bits $y_1,\ldots,y_n$ is 
symmetric if its output depends only on their sum.) In particular, ANDs, ORs, 
Modular sums and Majorities of subsets of the input alphabet all belong to this class. Until recently, 
Nisan's result gave the best known seedlength for any explicit $\varepsilon$-PRG for this class, 
even when 
$\varepsilon$ was a constant. In 2011, however, Gopalan et al. \cite{GMRZ} 
gave an explicit $\varepsilon$-PRG for this class with seedlength $O(\log(mn) + \log^2(1/\varepsilon))$. This seedlength is optimal as a function of $m$ and $n$ but suboptimal as a function of $\varepsilon$, and for the very interesting case of 
$\varepsilon = 1/n^{O(1)}$, this result does not improve upon Nisan's work.

Is the setting of small error important? We think the answer is yes, for many reasons. The first deals with the class of combinatorial shapes: 
many tests from this class accept a random input only with inverse polynomial probability 
(e.g., the alphabet is $\{0,1\}$ and the test accepts iff the Hamming weight of its $n$ input bits is 
$n/2$); for such tests, the guarantee that a $1/n^{o(1)}$-PRG 
gives us is unsatisfactory. Secondly, while designing PRGs for some class of statistical tests 
with (say) constant error, it often is the case that one needs PRGs with much smaller error --- e.g., one natural way of constructing almost-$\log 
n$ wise independent spaces uses PRGs that fool parity tests \cite{NN93} to within inverse 
polynomial error. Thirdly, the reason to improve the dependence on the error is simply because we 
know that such PRGs exist. Indeed, a randomly chosen function that expands $O(\log n)$ bits to an 
$n$-bit string is, w.h.p., an $\varepsilon$-PRG for $\varepsilon=1/\poly(n)$. Derandomizing this 
existence proof is a basic challenge in understanding how to eliminate randomness from 
existence proofs. The tools we gain in solving this problem might help us in solving others of a similar flavor.

\paragraph{\textbf{Our result.}} While we are unable to obtain optimal PRGs for 
the class of combinatorial shapes, we make progress on a well-studied weakening of this problem: the construction of 
an \emph{$\varepsilon$-Hitting Set} ($\varepsilon$-HS). An $\varepsilon$-HS for the class 
of combinatorial shapes has the property that any combinatorial shape that accepts at least an 
$\varepsilon$ fraction 
of truly random strings accepts at least one of the strings in the hitting set. This is clearly a 
weaker guarantee than what an $\varepsilon$-PRG gives us. Nevertheless, in many cases, this 
problem turns out to be very interesting and non-trivial: in particular, an $\varepsilon$-HS for 
the class of space-bounded computations would solve the long-standing open question of whether 
$\RL=\L$. Our 
main result is an explicit $\varepsilon$-HS of size $\poly(mn/\varepsilon)$ for the class of combinatorial shapes, which is \emph{optimal}, to within polynomial factors, 
for all errors.
\begin{theorem}[Main Result (informal)]
For any $m,n\in\naturals, \varepsilon > 0$, there is an explicit $\varepsilon$-HS for the class of 
combinatorial shapes of size $\poly(mn/\eps)$.
\end{theorem}
\paragraph{\textbf{Related work:}} There has been a substantial amount of research into both PRGs and hitting sets 
for many interesting subclasses of the class of combinatorial shapes, and also some 
generalizations. Naor and Naor \cite{NN93} constructed PRGs for parity tests of bits (alphabet size $2$); these results were 
extended by Lovett, Reingold, Trevisan, and Vadhan \cite{LRTV} and Meka and Zuckerman 
\cite{MekaZu09} to modular sums (with coefficients). Combinatorial rectangles, another subclass of combinatorial shapes, have also been the subject of much attention. A 
series of works \cite{EGLNV,ArmoniSaWiZh96,Lu02} have constructed  $\varepsilon$-PRGs for this 
class of functions: the best such PRG, due to Lu \cite{Lu02}, has seedlength $O(\log n + \log^{3/2}(1/\varepsilon))$. 
Linial, Luby, Saks, and Zuckerman \cite{LLSZ} constructed optimal hitting sets for this class of 
tests. We build on many ideas from this work.

We also mention two more recent results that are very pertinent to our work. The first is to do with Linear 
Threshold functions which are weighted generalizations of threshold symmetric functions of input bits. For this class,
Rabani and Shpilka \cite{RabaniShpilka} construct an explicit $\varepsilon$-HS of optimal size $\poly(n/\eps)$. They use a bucketing and expander walk 
construction to build their hitting set. Our construction uses similar ideas.

The final result that we use is the PRG for combinatorial shapes by Gopalan \etal~\cite{GMRZ} that was mentioned in the introduction. This work directly motivates our results and moreover, we use their PRG as a black-box within our construction.

%% file: prelims.tex
\begin{definition}[Combinatorial Shapes, Rectangles, Thresholds]
A function $ f $ is an \emph{$ (m,n) $-Combinatorial Shape} if there exist sets $ A_1,\ldots,A_n 
\subseteq [m]$ and a symmetric function $ h:\zeroone^n\to \zeroone $ such that $ f(x_1,\ldots,x_n) 
= h(1_{A_1}(x_1),\ldots$ $,1_{A_n}(x_n)) $.\footnote{$1_{A}$ is the indicator function of the set $A$.} If $h$ is the AND function, we call $f$ an \emph{$(m,n)$-Combinatorial Rectangle}. If $h$ is an unweighted threshold function (i.e. $h$ accepts iff $\sum_i 1_{A_i}(x_i) \ge \theta$ for some $\theta\in\naturals$), then $f$ is said to be an \emph{$(m,n)$-Combinatorial Threshold}. We denote by $ \cshape(m,n) $, $\crect(m,n)$, and $\cthr(m,n)$ the class of $(m,n)$-Combinatorial Shapes, Rectangles, and Thresholds respectively.
\end{definition}

\paragraph*{Notation.} In many arguments, we will work with a fixed collection of accepting sets $A_1,\ldots,A_n\subseteq [m]$ that will be clear from the context. In such a scenario, for $ i\in[n] $, we let $ X_i= 
1_{A_i}(x_i)$, $ p_i = |A_i|/m $, $ q_i = 1-p_i $ and $ w_i = p_iq_i $. Define the weight 
of a shape $ f $ as $ w(f)=\sum_i w_i $. For $\theta\in\naturals$, let $T^-_\theta$ (resp. $T^+_\theta$) be the function that accepts iff $\sum 
1_{A_i}(X_i)$ is at most (resp. at least) $\theta$. 
%and $T^+_\theta$ denote the function that accepts iff $\sum \one_{A_i}(x_i)\geq \theta$.

\begin{definition}[Pseudorandom Generators and Hitting Sets]
Let $ \mc{F}\subseteq \{0,1\}^D $ denote a boolean function family for some input domain $ D $. A function $ G:\zeroone^s\to D $ is an $\varepsilon$-pseudorandom generator ($\varepsilon$-PRG) with seedlength $ s $ for a class of functions $\mc{F}$ if for all $ f\in\mc{F} $,
\[
|\Pr_{x\in_u\zeroone^s}[f(G(x))=1]- \Pr_{y\in_u D}[f(y)=1]|\leq \eps.
\]
An $ \eps $-hitting set ($\eps$-HS) for $ \mc{F} $ is a multiset $ H$ containing only elements from $ D $ s.t. for any $ f\in \mc{F} $, if 
$ \Pr_{x\in_u D}[f(x)=1] \geq \eps$, then $\exists x\in H$ s.t. $f(x)=1$. 
\end{definition}

\begin{remark}
Whenever we say that there exist \emph{explicit} families of combinatorial objects of some kind, we mean that the object can be constructed by a deterministic algorithm in time polynomial in the description of the object. 
\end{remark}
%
%\begin{definition}[Pseudorandom Generators]
%A function $ G:\zeroone^s\to [m]^n $ is a pseudorandom generator (PRG) with seedlength $ s $ and 
%error $ \eps $ for a class of functions $ \mc{F}:[m]^n \to \zeroone$ if for all $ f\in\mc{F} $,
%\[
%\abs{\Pr_{x\in_u\zeroone^s}[f(G(x))=1]-\Pr_{y\in_u [m]^n}[f(y)=1]}\leq \eps.
%\]
%\end{definition}

We will need the following previous results in our constructions.

\begin{theorem}[$\eps$-PRGs for $\cshape(m,n)$ \cite{GMRZ}]
\label{thm:GMRZ}
For every $ \eps > 0 $, there exists an explicit $\eps$-PRG $\mc{G}_{GMRZ}^{m,n,\eps}:\{0,1\}^s\rightarrow [m]^n$ for $ \cshape(m,n) $ with seed-length $s = O(\log(mn) + \log^2(1/\eps)) $.
\end{theorem}
%
%\begin{definition}[Combinatorial Rectangles]
%A function $ f:[m]^n\to \zeroone $ is an $ (m,n) $-combinatorial rectangle if there exist sets $ 
%A_1,\ldots,A_n \subseteq [m] $ such that $ f(x_1,\ldots,x_n)\equiv \bigwedge_{i=1}^{n} 
%1_{A_i}(x_i) $. 
%Let $ \crect(m,n) $ denote the class of such functions.
%\end{definition}
%
%Linial \etal~\cite{LLSZ} gives the optimal hitting set for combinatorial rectangles.
%

\begin{theorem} [$\eps$-$\hs$ for $\crect(m,n)$ \cite{LLSZ}]
	\label{th:LLSZ}
For every $ \eps > 0 $, there exists an explicit $ \eps $-hitting set $ \mc{S}_{LLSZ}^{m,n,\eps} $ for $ 
\crect(m,n) $ of size $ 
\poly(m(\log n)/\eps) $.
\end{theorem}

We will also need a stronger version of Theorem~\ref{th:LLSZ} for special cases of combinatorial rectangles. Informally, the strengthening says that if the acceptance probability of a `nice' rectangle is $>p$ for some {\em reasonably large} $p$, then a close to $p$ fraction of the strings in the hitting set are accepting. Formally, the following is proved later in the paper.
\begin{theorem} [Stronger $\hs$ for $\crect(m,n)$]
	\label{thm:comb-rect}
For all constants $c\geq 1$, $m=n^c$, and $\rho \le c\log n$, there is an
explicit set $\calS_{\text{rect}}^{n,c,\rho}$ of size $n^{O_c(1)}$
s.t. for any $\calR \in \crect(m, n)$ which satisfies the properties:
\begin{enumerate}
\item $\calR$ is defined by $A_i$, and the rejecting probabilities
$q_i:= (1-|A_i|/m)$ which satisfy $\sum_i q_i \le \rho$,
\item $\Pr_{X\sim [m]^n} [ \calR(X) = 1] \ge p ~(\ge 1/n^c)$
\end{enumerate}
we have \[ \Pr_{X\sim \calS_{\text{rect}}^{n,c,\rho}} [ \calR(X) = 1] \ge \frac{p}{2^{O_c(\rho)}}.\]
\end{theorem}
%
%\begin{definition}[Combinatorial Thresholds]
%	A function $ f:[m]^n\to \zeroone $ is an $ (m,n) $-combinatorial threshold if there exist sets 
%$ 
%	A_1,\ldots,A_n \subseteq [m] $ and a threshold $ \theta $ such that $ f(x_1,\ldots,x_n)$ is $ 
%0/1 $ depending on the sign of $ \theta - \sum_{i=1}^{n} 1_{A_i}(x_i) $. 
%	Let $ \cthr(m,n) $ denote the class of such functions.
%\end{definition}
%
%\begin{remark}
%For the results in this paper we can assume w.l.o.g. that the parameters $ m$,$n $ and $ \eps $ 
%are all polynomially related, that is, $ m=n^{O(1)} $ and $ \eps=1/n^{O(1)} $ (see 
%Lemma~\ref{lem:params} for the justification).
%\end{remark}

Recall that a distribution $\mu$ over $[m]^n$ is $k$-wise independent for $k\in\naturals$ if for any $S\subseteq [n]$ s.t. $|S|\leq k$, the marginal $\mu|_S$ is uniform over $[m]^{|S|}$. Also, $\mc{G}:\{0,1\}^s \rightarrow [m]^n$ is a 
\emph{$k$-wise independent probability space over $[m]^n$} if for  uniformly randomly chosen $z\in \{0,1\}^s$, the distribution of 
$\mc{G}(z)$ is $k$-wise independent.
%	
%	We will also need the notion of $ k $-wise independent spaces.
%	
%\begin{definition}[$k$-wise independent spaces]
%\label{defn_k_wise}
%Fix $k\in\naturals$. A function $\mc{G}:\{0,1\}^s \rightarrow [m]^n$ is said to generate a 
%\emph{$k$-wise independent probability space over $[m]^n$} if for any non-empty subset $S\subseteq 
%[n]$ s.t. $|S|\leq k$ and uniformly randomly chosen $z\in \{0,1\}^s$, the random variable 
%$\mc{G}(z)|_S$ (that is, the vector obtained by restricting $\mc{G}(z)$ to the indices in S) is 
%uniformly distributed over $[m]^{|S|}$. $ s $ is known as the seed-length used by the generator.
%\end{definition}
%
\begin{fact}[Explicit $k$-wise independent spaces]
\label{fact_k_wise_existence}
For any $k,m,n\in\naturals$, there is an explicit $k$-wise independent probability space 
$\mc{G}^{m,n}_{k\text{-wise}}:\{0,1\}^s\rightarrow [m]^n$ with $s = O(k\log (mn))$.
\end{fact}

We will also use the following result of Even et al.~\cite{EGLNV}.
\begin{theorem}
\label{thm:EGLNV}
Fix any $m,n,k\in\naturals$. Then, if $f\in \crect(m,n)$ and $\mu$ is any $k$-wise independent distribution over $[m]^n$, then  we have
\[
\left|
\prob{x\in [m]^n}{f(x)=1} - \prob{x\sim \mu}{f(x)=1}
\right| \leq \frac{1}{2^{\Omega(k)}}
\]
\end{theorem}

\paragraph*{Expanders.}
Recall that a degree-$D$ multigraph $G = (V,E)$ on $N$ vertices is an 
$(N,D,\lambda)$-expander if the second largest (in absolute value) eigenvalue of its normalized adjacency matrix is at 
most $\lambda$. 
We will use explicit expanders as a basic building block. We refer the reader to the excellent survey of Hoory, Linial, and Wigderson \cite{HLW} for various related results.

\begin{fact}[Explicit Expanders \cite{HLW}]
\label{fact_expl_exp}
Given any $\lambda > 0$ and $N\in\naturals$, there is an explicit $(N,D,\lambda)$-expander where 
$D = (1/\lambda)^{O(1)}$.
\end{fact}

Expanders have found numerous applications in derandomization. A central theme in these 
applications is to analyze random walks on a sequence of expander graphs. Let $G_1,\ldots,G_{\ell}$ be a sequence of (possibly different) graphs on the \emph{same} vertex set $V$. Assume $G_i$ ($i\in [\ell]$) is an $(N,D_i,\lambda_i)$-expander. Fix any $u\in V$ and $y_1,\ldots,y_\ell\in\naturals$ s.t. $y_i\in [D_i]$ for each $i\in [\ell]$. Note that $(u,y_1,\ldots,y_\ell)$ naturally defines a `walk' $(v_1,\ldots,v_\ell)\in V^{\ell}$ as follows: $v_1$ is the $y_1$th neighbour of $u$ in $G_1$ and for each $i>1$, $v_{i}$ is the $y_i$th neighbour of $v_{i-1}$ in $G_i$. We denote by $\mc{W}(G_1,\ldots,G_\ell)$ the set of all tuples $(u,y_1,\ldots,y_\ell)$ as defined above. Moreover, given $w=(u,y_1,\ldots,y_\ell)\in\mc{W}(G_1,\ldots,G_{\ell})$, we define $v_i(w)$  to be the vertex $v_i$ defined above (we will simply use $v_i$ if the walk $w$ is clear from the context).

We need a variant of a result due to Alon, Feige, Wigderson, and Zuckerman \cite{AFWZ}. The lemma as it is stated below is slightly more general than the one given in \cite{AFWZ} but it can be obtained by using essentially the same proof and setting the parameters appropriately. The proof is given in the appendix.
\begin{lemma}
\label{lemma_afwz}
Let $G_1,\ldots,G_\ell$ be a sequence of graphs defined on the same vertex set $V$ of size $N$. Assume that $G_i$ is an $(N,D_i,\lambda_i)$-expander. Let $V_1,\ldots,V_\ell\subseteq V$ s.t. $|V_i|\geq p_i N > 0$ for each $i\in [\ell]$. Then, as long as for each $i \in [\ell]$, $\lambda_i \leq (p_{i}p_{i-1})/8$,
\begin{equation}\label{eqn:afwz}
\prob{w\in \mc{W}(G_1,\ldots,G_\ell)}{\forall i\in [\ell], v_i(w)\in V_i} \geq (0.75)^\ell \prod_{i\in [\ell]}{p_i}.
\end{equation}
\end{lemma}

Also, in our applications, we will sometimes use the following corollary.

\begin{corollary}
Let $V$ be a set of $N$ elements, and let $0 < p_i<1$ for $1 \le i\le s$ be given.  There exists an explicit set of walks $\calW$, each of length $s$, s.t. for any subsets $V_1, V_2, \dots, V_s$ of $V$, with $|V_i| \ge p_i N$, there exists a walk $w=w_1w_2 \dots w_s \in \calW$ such that $w_i \in V_i$ for all $i$. Furthermore, there exist such $\calW$ satisfying
$ |\calW| \le \text{poly}\big( N, \prod_{i=1}^s \frac{1}{p_i} \big)$. 
\end{corollary}
This follows from Lemma~\ref{lemma_afwz} by picking $\lambda_i$ smaller than $p_i p_{i-1}/10$ for each $i$.  By Fact \ref{fact_expl_exp}, known explicit constructions of expanders require choosing degrees $d_i = 1/\lambda_i^{O(1)}$.  The number of walks of length $s$ is $N\cdot \prod_{i=1}^\ell d_i$, which gives the bound on $\calW$ above.
%%
%\paragraph{A Note on Constants.} Our constructions inherently involve large yet absolute constants. The letters $C, c, c_{walk}$ etc. are
%used to denote them. We do not strictly keep track of them when not necessary.
%

\paragraph{Hashing.} Hashing plays a vital role in all our constructions. Thus, we need explicit hash families which have several ``good'' properties.  First, we state a lemma obtained by slightly extending part of a lemma due to Rabani and Shpilka 
\cite{RabaniShpilka}, which itself builds on the work of Schmidt and Siegel \cite{SchmidtSiegel} 
and Fredman, Koml\'{o}s, and Szemer\'{e}di \cite{FKS}. It is somewhat folklore and the proof is omitted.

\begin{lemma}[Perfect Hash Families]
	\label{lemma_perfect_hash}
	%There is an absolute constant $c_{perf}>0$ so that the following holds.
	For any $n,t\in\naturals$, there is an explicit family of hash functions 
	$\mc{H}_{\text{perf}}^{n,t}\subseteq [t]^{[n]}$ of size $2^{O(t)}\poly(n)$ s.t. for any $S\subseteq 
	[n]$ with $|S|=t$, we have
	\[
	\prob{h\in \mc{H}_{\text{perf}}^{n,t}}{\text{$h$ is $1$-$1$ on $S$}}\geq \frac{1}{2^{O(t)}}.
	\]
\end{lemma}

The family of functions thus constructed are called ``perfect hash families''.  We also need a {\em fractional}
version of the above lemma, whose proof is
similar to that of the perfect hashing lemma and is presented later in the paper.

\begin{lemma}[Fractional Perfect Hash families]\label{lem:fract-hash}
For any $n,t\in\naturals$, there is an explicit family of hash functions $\mc{H}_{frac}^{n,t}\subseteq [t]^{[n]}$ of size $2^{O(t)}n^{O(1)}$ s.t. for any $z\in [0,1]^n$ with $\sum_{j\in [n]} z_j \geq 10t$, we have
\[
\prob{h\in \mc{H}_{frac}^{n,t}}{\forall i\in [t], \sum_{j\in 
h^{-1}(i)}z_j \in [0.01M,10M]}\geq \frac{1}{2^{O(t)}},
\]
where $M=\frac{\sum_{j\in [n]} z_j}{t}$.
\end{lemma}

%% file: outline.tex
We will outline some simplifying assumptions, and an observation which
``reduces'' constructing hitting sets for Combinatorial shapes
$\cshape(m,n)$ to those for Combinatorial Thresholds $\cthr(m,n)$. It
turns out that these are somewhat simpler to construct, appealing to
the recent results of Gopalan \etal~\cite{GMRZ}.

We first make a standard simplifying observation that we can throughout assume that $m, n, 1/\eps$ can be $n^{O(1)}$. Thus, we only need to construct hitting sets of size $n^{O(1)}$ in this case. From now on, we assume $m,1/\varepsilon = n^{O(1)}$.
\begin{lemma}\label{lem:params}
Assume that for some $c\geq 1$, and $m\leq n^c$, there is an explicit $1/n^c$-HS for $\cshape(m,n)$ of size $n^{O_c(1)}$. Then, for any $m,n,\in\naturals$ and $\varepsilon>0$, there is an explicit $\varepsilon$-HS for $\cshape(m,n)$ of size $\poly(mn/\varepsilon)$.
\end{lemma}
\begin{proof}
Fix $c\geq 1$ so that the assumptions of the lemma hold. Note that when $ m > n^c $, we can increase the number of coordinates to $ n' = 
m $. Now, an $ \eps $-HS for $ \cshape(m,n') $ is also an $ \eps $-HS for $ \cshape(m,n) $, 
because we can ignore the final $ n'-n $ coordinates and this will not affect the hitting set 
property. 
Similarly, when $ \eps < 1/n^c $, we can again increase the number of coordinates to $ n' $ that 
satisfies $ \eps \geq 1/(n')^c $ and the same argument follows. In each case, by assumption we have an $\eps$-HS of size $(n')^{O_c(1)} = \poly(mn/\eps)$ and thus, the lemma follows.
\end{proof}

Next, we prove a crucial lemma which shows how to obtain hitting sets for $\cshape(m,n)$
starting with hitting sets for $\cthr(m,n)$.  This reduction crucially uses the fact that $\cshape$ does only `symmetric' tests -- it fails to hold, for instance, for natural ``weighted'' generalizations of $\cshape$.
\begin{lemma}\label{lem:cthr_suff}
Suppose that for every $\eps>0$, there exist an explicit $\eps$-HS for $\cthr(m,n)$ of size $F(m, n, 1/\eps)$.  Then there exists an explicit $\eps$-HS for $\cshape(m,n)$ of size $(n+1) \cdot
F^2(m, n, (n+1)/\eps)$.
\end{lemma}
\begin{proof}
Suppose we can construct hitting sets for $\cthr(m,n)$ and parameter
$\eps'$ of size $F(m, n, 1/\eps')$, for all $\eps'>0$. Now consider
some $f \in \cshape(m,n)$, defined using sets $A_i$ and symmetric
function $h$. Since $h$ is symmetric, it depends only on the {\em
  number} of $1$'s in its input. In particular, there is a $W\subseteq [n]\cup\{0\}$ s.t. for
$a\in \{0,1\}^n$ we have $h(a)=1$ iff $|a|\in W$. Now if $\Pr_x[f(x)=1] \ge \eps$, there must exist a $w
\in W$ s.t. 
\[
\prob{x}{|\{i\in [n]\ |\ 1_{A_i}(x_i)=1\}| = w} \ge \frac{\eps}{|W|} \ge \frac{\eps}{n+1}.
\]
Thus if we consider functions in $\cthr(m,n)$ defined by the same
$A_i$, and thresholds $T^+_w$ and $T^-_w$ respectively, we
have that {\em both} have accepting probability at least $\eps/(n+1)$, and
thus an $ \epsilon/(n+1) $-HS $\calS$ for $\cthr(m, n)$ must
have `accepting' elements $y, z \in [m]^n$ for $T^-_w$ and $T^+_w$
respectively.

The key idea is now the following. Suppose we started with the string
$y$ and moved to string $z$ by flipping the coordinates one at a time
-- i.e., the sequence of strings would be:
\[
(y_1~y_2 ~\dots ~ y_n),
~(z_1~y_2 ~\dots ~ y_n), ~(z_1~z_2 ~\dots ~ y_n),\dots, (z_1~z_2 ~\dots ~ z_n).
\]

In this sequence the number of ``accepted'' indices (i.e., $i$ for
which $1_{A_i}(x_i) =1$) changes by at most one in each `step'. To
start with, since $y$ was accepting for $T^-_w$, the number of
accepting indices was at most $w$, and in the end, the number is at least
$w$ (since $z$ is accepting for $T^+_w$), and hence one of the
strings must have precisely $w$ accepting indices, and this string
would be accepting for $f$!

Thus, we can construct an $\eps$-HS for $\cshape(m,n)$ as follows. Let $\mc{S}$ denote an explicit $(\eps/(n+1))$-HS for $\cthr(m,n)$ of size $F(m,n,1/\eps)$. For any $y,z\in\mc{S}$, let $\mc{I}_{y,z}$ be the set of $n+1$ ``interpolated'' strings obtained above. Define $\mc{S}' = \bigcup_{y,z\in \mc{S}} \mc{I}_{y,z}$. As we have argued above, $\mc{S}'$ is an $\eps$-HS for $\cshape(m,n)$. It is easy to check that $\mc{S}'$ has the size claimed.
%
%Thus we can consider every pair of strings in a hitting set for
%$\cthr(m,n)$ and error $\eps/(n+1)$, and consider the $(n+1)$
%``intermediate'' strings as a hitting set for $\cshape(m,n)$ of error
%$\eps$. It is easy to check that it has size $(n+1) \cdot
%F^2(m,n,(n+1)/\eps)$.
\end{proof}

\paragraph{Overview of the Constructions.} In what follows, we focus on constructing
hitting sets for $\cthr(m,n)$. 
We will describe the construction of two families of hitting sets: the first is for the ``high weight'' case -- $w(f):= \sum_i w_i > C\log n$ for some large constant $C$, and the second for the case $w(f) < C\log n$.  The final hitting set is a union of the ones for the two cases.  

The high-weight case (Section~\ref{sec:highwt}) is conceptually simpler, and illustrates the important tools.  A main tool in both cases is a ``fractional'' version of the perfect hashing lemma, which, though a consequence of folklore techniques, does not seem to be known in this generality (Lemma~\ref{lem:fract-hash}). 

The proof of the low-weight case is technically more involved, so we first present the solution in the special case when all the sets $A_i$ are ``small'', i.e., we have $p_i \leq 1/2$ for all $i$ (Section~\ref{sec:lowwt_lowsize}). This case illustrates the main techniques we use for the general low-weight case. The special case uses the perfect hashing lemma (which appears, for instance in derandomization of ``color coding'' -- a trick introduced in \cite{AYZ}, which our proof in fact bears a resemblance to). 

 The general case (Section~\ref{sec:lowwt_general}), in which $p_i$ are arbitrary, is more technical: here we need to do a ``two level'' hashing. The top level is by dividing into buckets, and in each bucket we get the desired ``advantage'' using a generalization of hitting sets for combinatorial rectangles (which itself uses hashing: Theorem~\ref{thm:comb-rect}).

Finally we describe the main tools used in our construction. The stronger hitting set construction for special combinatorial rectangles is discussed in Section~\ref{sec:comb-rectangles} and the fractional perfect hash family construction is discussed in Section~\ref{sec:fract-hash}. We end with some interesting open problems and a proof of the expander walk lemma follows in the appendix.

%% file: highwt.tex
\subsection{High weight case}
 \label{sec:highwt}

In this section we will prove the following:

\begin{theorem}\label{thm_high_wt}
For any $c\geq 1$, there is a $C>0$ s.t. for $m,1/\varepsilon\leq n^c$, there is an explicit $\varepsilon$-HS of size $n^{O_c(1)}$ for the class of functions in $ \cthr(m,n)$ of weight at least $C\log n$.
\end{theorem}
As discussed earlier, we wish to construct hitting sets for combinatorial shapes $f$ where the associated symmetric function is either 
$T^+_{\theta}$ or $T^-_{\theta}$, for $\theta$ s.t. the
probability of the event for independent, perfectly random $x_i$ is at
least $1/n^c$. For convenience, define $\mu := p_1 + p_2 +
\dots + p_n$, and $W := w_1 + w_2 + \dots w_n$. We have $W > C\log n$
for a large constant $C$ (it needs to be {\em large} compared to $c$,
as seen below). First, we have the following by Chernoff bounds.

\begin{claim}
If $\Pr_x[T^+_{\theta}(\sum_{i\in [n]}1_{A_i}(x_i))=1] > \eps~ (\ge 1/n^c)$, we have $\theta \le \mu +
2\sqrt{cW\log
  n}$.
\end{claim}

\paragraph{Outline.} Let us concentrate on hitting sets for combinatorial shapes that use symmetric functions of the form
$T^+_{\theta}$ (the
case $T^-_{\theta}$ follows verbatim).  The main idea is
the following: we first divide the indices $[n]$ into $\log n$ buckets
using a hash function $h$ (from a {\em fractional
  perfect hash family}, see Lemma~\ref{lem:fract-hash}). This is to
ensure that the $w_i$ get distributed somewhat uniformly. Second, we aim to
obtain an {\em advantage} of roughly $2\sqrt{ \frac{cW}{\log n}}$ in
each of the buckets (advantage is w.r.t. the mean in each
bucket): i.e., for each $i\in [\log n]$, we choose the indices $x_j$ ($j\in h^{-1}(i)$) s.t. we get 
$$\sum_{j\in h^{-1}(i)}1_{A_j}(x_j) \geq \sum_{j\in h^{-1}(i)}p_j + 2\sqrt{cW/\log n}$$

with reasonable probability. Third, we ensure that the above happens for all buckets \emph{simultaneously} (with probability $>0$) so that the advantages add up, giving a total
advantage of $2\sqrt{cW\log n}$ over the mean, which is what we
intended to obtain. In the second step (i.e., in each bucket), we
can prove that the desired advantage occurs with {\em constant} probability for \emph{uniformly randomly and independently} chosen $x_j\in [m]$ and 
then derandomize this choice by the result of Gopalan \etal~\cite{GMRZ} (Theorem \ref{thm:GMRZ}). Finally, in the third step, we
cannot afford
to use independent random bits in different buckets
(this would result in a seed length of $\Theta(\log^2 n)$) -- thus we
need to use expander walks to save on randomness. 

\paragraph{Construction and Analysis.} Let us now describe
the three steps in detail. We note that these steps parallel the results 
of Rabani and Shpilka~\cite{RabaniShpilka}.

The first step is straightforward: we pick a hash function from a perfect fractional
hash family $\calH_{\text{frac}}^{n, \log n}$. From
Lemma~\ref{lem:fract-hash}, we obtain
\begin{claim}\label{claim:hash-result}
For every set of weights $w$, there exists an $h \in \calH^{n, \log
n}_{\text{frac}}$ s.t. for
all $1 \le i \le \log n$, we have $\frac{W}{100\log n} \le \sum_{j \in
h^{-1}(i)} w_j \le \frac{100 W}{\log n}$.
\end{claim}

The rest of the construction is done starting with each $h \in
\calH_{\text{frac}}^{n, \log n}$.
Thus for analysis, suppose that we are working with an $h$
satisfying the inequality from the above claim. For the second step, we first
prove
that for independent random $x_i \in [m]$, we have a constant
probability of getting an {\em advantage} of $2\sqrt{\frac{cW}{\log
    n}}$ over the mean in each bucket.

\begin{lemma}\label{lem:berry-ess2}
Let $S$ be the sum of $k$ independent random variables $X_i$, with
$\Pr[X_i=1] =p_i$, let $c'\geq 0$ be a constant, and let $\sum_i p_i(1-p_i) \ge
\sigma^2$, for some $\sigma$ satisfying $\sigma \ge 20e^{c'^2}$. Define $\mu :=
\sum_i p_i$. Then $\Pr[S > \mu + c'\sigma] \ge \alpha$, and
$\Pr[S < \mu - c'\sigma] \ge \alpha$, for some constant $\alpha$ depending on
$c'$.
\end{lemma}
The proof is straightforward, but
it is instructive to note that in general, a random variable (in this
case, $S$) need not deviate ``much more'' (in this case, a $c'$ factor more)
than its standard deviation:
we have to use the fact that $S$ is the sum of independent r.v.s. This is
done by an application of the Berry-Ess\'een theorem~\cite{probability-text}.

\begin{proof}
 We recall the standard Berry-Ess\'een theorem~\cite{probability-text}.
\begin{fact}[Berry-Esseen]
 \label{fact_BE}
Let $Y_1,\dots,Y_n$ be independent random variables satisfying
$\forall i$, $\E Y_i = 0$, $\sum\E Y_i^2 = \sigma^2$ and $\forall i$,
$|Y_i|\leq\beta\sigma$. Then the following error bound holds for any $t\in \R$,
\[
 \left|\Pr\left[\sum Y_i > t\right] - \Pr\left[N(0,\sigma^2) > t\right]\right| \leq \beta.
\]
\end{fact}
We can now apply this to $Y_i := X_i-p_i$ (so as to make $\E
Y_i=0$). Then $\E Y_i^2 = p_i(1-p_i)^2 + (1-p_i)p_i^2 = p_i(1-p_i)$,
thus the total variance is still $\ge \sigma^2$. Since $|Y_i|\leq 1$ for all $i\in [n]$, this means we have the
condition $|Y_i|\le \beta \sigma$ for $\beta \le e^{-c'^2}/20$. Now for the
Gaussian, a computation shows that we have $\Pr[ N(0,\sigma^2) > c'\sigma] >
e^{-c'^2}/10$.  Thus from our bound on $\beta$, we get $\Pr[\sum Y_i > c'\sigma] >
e^{-c'^2}/20$, which we pick to be $\alpha$. This proves the lemma.
\end{proof}

Assume now that we choose $x_1,\ldots,x_n\in [m]$ independently and uniformly at random. For each bucket $i\in [\log n]$ defined by the hash function $h$, we let $\mu_i = \sum_{j\in h^{-1}(i)}p_j$ and $W_i = \sum_{j\in h^{-1}(i)}p_j(1-p_j) = \sum_{j\in h^{-1}(i)}w_j$. Recall that Claim \ref{claim:hash-result} assures us that for $i\in [\log n]$, $W_i \geq W/100\log n \geq C/100$. Let $X^{(i)}$ denote $\sum_{j\in h^{-1}(i)}1_{A_j}(x_j)$. Then, for any $i\in [\log n]$, we have
\[
\prob{}{X^{(i)} > \mu_i + 2\sqrt{\frac{cW}{\log n}}} \geq \prob{}{X^{(i)} > \mu_i + \sqrt{400c}\cdot\sqrt{W_i}}
\]

By Lemma \ref{lem:berry-ess2}, if $C$ is a large enough constant so that $W_i \geq C/100 \geq 20 e^{400c}$, then for uniformly randomly chosen $x_1,\ldots,x_n\in [m]$ and each bucket $i\in [\log n]$, we have $\prob{}{X^{(i)} \geq \mu_i+ 2\sqrt{cW/\log n}}\geq \alpha$, where $\alpha >0$ is some fixed constant depending on $c$. When this event occurs for \emph{every} bucket, we obtain $\sum_{j\in [n]}1_{A_j}(x_j) \geq \mu + 2\sqrt{cW\log n} \geq \mu + \theta$. We now show how to sample such an $x\in [m]^n$ with a small number of random bits.

Let $\mc{G}:\{0,1\}^s\rightarrow [m]^n$ denote the PRG  of Gopalan et al.~\cite{GMRZ} from Theorem \ref{thm:GMRZ} with parameters $m,n,$ and error $\alpha/2$ i.e. $\mc{G}_{GMRZ}^{m,n,\alpha/2}$. Note that since $\alpha$ is a constant depending on $c$, we have $s = O_c(\log n)$. Moreover, since we know that the success probability with independent random $x_j$ ($j\in h^{-1}(i)$) for obtaining the desired advantage is at least $\alpha$, we have for any $i\in [\log n]$ and $y^{(i)}$ randomly chosen from $\{0,1\}^s$,
\[
\prob{x^{(i)} = \mc{G}(y^{(i)})}{X^{(i)} > \mu_i + 2\sqrt{\frac{cW}{\log n}}} \geq \alpha/2
\]
This only requires seedlength $O_c(\log n)$ per bucket.

Thus we are left with the third step: here for each bucket $i\in [\log n]$, we would like to
have (independent) seeds which generate the corresponding $x^{(i)}$ (and each of these PRGs
has a seed length of $O_c(\log n)$). Since we cannot afford $O_c(\log^2 n)$ total
seed length, we instead do
the following: consider the PRG $\mc{G}$ defined above.  As mentioned above, since $\alpha = \Omega_c(1)$,
the seed length needed here 
is only $O_c(\log n)$. Let $\calS$ be the range of $\mc{G}$ (viewed as a multiset of
strings: $\calS \subseteq [m]^n$).  From
the above, we have that for the $i$th bucket, the probability $x \sim
\calS$ exceeds the threshold on indices in bucket $i$ is at least
$\alpha/2$.
Now there are $\log n$ buckets, and in each bucket, the probability of
`success' is at least $\alpha/2$.  We can
thus appeal to the `expander walk' lemma of Alon \etal~\cite{AFWZ} (see
preliminaries, Lemma~\ref{lemma_afwz} and the corollary following it).

This means the following: we consider an explicitly constructed expander
on a graph with vertices being the elements of $\calS$, and the degree being
a constant depending on $\alpha$). We then perform a
random walk of length $\log n$ (the number of buckets). Let $s_1, s_2, \dots,
s_{\log n}$ be the strings (from $\calS$) we see in the walk. We form a new string in $[m]^n$
by picking values for indices in bucket $i$, from the string $s_i$. By the
Lemma \ref{lemma_afwz}, with non-zero probability, this
will succeed for {\em all} $1 \le i \le \log n$, and this gives the desired
advantage.

The seed length for generating the walk is $O(\log |\calS|)+O_c(1) \cdot \log n =
O_c(\log n)$. Combining (or
in some sense, {\em composing}) this with the hashing earlier completes
the construction.

%% file: lowwt-lowsize.tex
\subsection{Thresholds with small weight and small sized sets}
\label{sec:lowwt_lowsize}

We now prove Theorem~\ref{thm:threshold-fool} for the case of thresholds $ f $
satisfying $w(f) = O(\log n)$. Also we will make the simplifying assumption
(which we will get rid of in the next sub-section) that the underlying subsets
of $f$, $A_1,\ldots, A_n\subseteq[m]$ are of small size.

\begin{theorem} \label{thm:low_wt_low_size} Fix any $c \ge 1$. For any $m =
n^c$, there exists an explicit $1/n^c$-HS $\mathcal{S}_\text{low,1}^{n,c}
\subseteq [m]^n$ of size $n^{O_{c}(1)}$ for  functions $f \in \cthr(m,n)$ s.t.
$w(f)\leq c\log n$ and $p_i\leq 1/2$ for each $i \in [n]$. \end{theorem}

We will prove this theorem in the rest of this sub-section. Note that since $p_i\leq 1/2$ for each $i\in [n]$, we have $w_i = p_i(1-p_i) \geq p_i/2$.

To begin, we note
that hitting sets for the case when the symmetric function is $T^-_{\theta}$ is easily obtained.  In particular, since $T^-_0$ accepts iff $\sum
X_i=0$, it can also be interpreted as a combinatorial rectangle with accepting sets 
$\bar{A_1},\ldots,\bar{A_n}$. The probability of this event over uniformly
chosen inputs is at least $ \prod_i (1-p_i) \ge e^{-2\sum_i p_i} \ge e^{-4\sum_i
p_i (1-p_i)} \ge n^{-4c}$, where the first inequality uses the fact that $(1-x)\geq e^{-2x}$ for $x\in [0,1/2]$.  Thus the existence of a hitting set for such $f$ follows from the result of Linial
\etal.~\cite{LLSZ}.  Further, by definition, a hitting set for $T^-_0$ is also a
hitting set for $T^-_\theta$ for $\theta >0$. We will therefore focus on hitting sets for thresholds of the form $T^+_\theta$ for some $\theta>0$.

Let us now fix a function $f(x) = T^+_\theta(\sum_i 1_{A_i}(x_i))$ that accepts with good probability: $\Pr_x
[T^+_\theta(\sum_i 1_{A_i}(x_i))=1] \ge \eps$. Since $w(f)\leq c\log n$ and $p_i \leq 2w_i$ for each $i\in [n]$, it follows that
$ \mu\leq 2c\log n $. Thus by a Chernoff bound and the fact that $\eps = 1/n^c$,
we have that $\theta \le c'\log n$ for some $c' = O_c(1)$.

\paragraph{Outline.} 

The idea is to use a hash function $h$ from a {\em perfect hash family} (Lemma \ref{lemma_perfect_hash}) mapping $[n]\mapsto [\theta]$. The aim will now be to obtain a
contribution of $1$ to the sum $\sum_i 1_{A_i}(x_i)$ from each bucket\footnote{This differs from the high-weight
case, where we looked at advantage over the mean.}. In order to do this, we
require $\prod_i \mu_i$ be large, where $\mu_i$ is the sum of $p_j$ for $j$ in
bucket $B_i = h^{-1}(i)$. By a reason similar to color coding (see~\cite{AYZ}), it will turn
out that this quantity is large when we bucket using a perfect hash family. We
then prove that using a pairwise independent space in each bucket $B_i$
``nearly'' gives probability $\mu_i$ of succeeding. As before, since we cannot
use independent hashes in each bucket, we take a hash function over $[n]$, and
do an expander walk. The final twist is that in the expander walk, we cannot use
a constant degree expander: we will have to use a sequence of expanders on the
same vertex set with appropriate degrees (some of which can be super-constant,
but the product will be small). This will complete the proof. We note that the
last trick was implicitly used in the work of \cite{LLSZ}.

\paragraph{Construction.}

Let us formally describe a hitting set for $T^+_\theta$ for a fixed $\theta$.
(The final set $\mathcal{S}_\text{low,1}^{n,c}$ will be a union of these for
$\theta \le c'\log n$ along with the hitting set of \cite{LLSZ}).

\textbf{Step 1:} Let $\mc{H}^{n,\theta}_{perf}=\{h:[n]\to[\theta]\} $ be a
perfect hash family as in Lemma~\ref{lemma_perfect_hash}. The size of the hash
family is $ 2^{O(\theta)}\poly(n) = n^{O_{c'}(1)} = n^{O_c(1)}$. For each hash
function $ h\in \mc{H}^{n,\theta}_{perf} $ divide $[n] $ into $ \theta $ buckets
$ B_1,\ldots,B_\theta $ (so $ B_i=h^{-1}(i) $).

\textbf{Step 2:} We will plug in a pairwise independent space in each bucket.
Let $ \mc{G}^{m,n}_{2-wise}:\zeroone^s\to[m]^n $ denote the generator of a
pairwise independent space. Note that the seed-length for any bucket is $
s=O(\log n)$\footnote{We do not use generators with different output lengths,
instead we take the $n$-bit output of one generator and restrict to the entries
in each bucket.}.

\textbf{Step 3:} The seed for the first bucket is chosen uniformly at random and
seeds for the subsequent buckets are chosen by a walk on expanders with varying
degrees. For each $ i\in[\theta] $ we choose every possible $ \eta'_i $ such
that $1/\eta'_i $ is a power of $ 2 $ and $\prod_i\eta'_i \geq 1/n^{O_c(1)}$, where the constant implicit in the $O_c(1)$ will become clear in the analysis of the construction below. There are
at most $ \poly(n) $ such choices for all $ \eta'_i $'s in total. We then take
a $ (2^s,d_i,\lambda_i) $-expander $H_i $ on vertices $ \zeroone^s $ with
degree $ d_i=\poly(1/(\eta'_i\eta'_{i-1})) $ and $ \lambda_i \leq
\eta'_i\eta'_{i-1}/100$ (by Fact \ref{fact_expl_exp}, such explicit expanders exist). Now for any $ u\in\zeroone^s $, $
\{y_i\in[d_i]\}_{i=1}^\theta $, let $(u,y_1,\ldots,y_\theta)\in\mc{W}(H_1,\ldots,H_\theta) $ be
a $ \theta $-step walk. For all starting seeds $ z_0 \in \zeroone^s $ and all
possible $ y_i \in [d_i] $, we construct the input $x\in [m]^n$ s.t. for all $ i\in[\theta] $, we have $x|_{B_i} = \mc{G}^{m,n}_{2-wise}(v_i(z_0,y_1,\ldots,y_\theta))|_{B_i}$.

\textit{Size.} We have $|\mathcal{S}_\text{low,1}^{n,c}|=c'\log n \cdot n^{O_c(1)}
\cdot {\prod_i d_i}$, where the $ c'\log n $ factor is due to the choice of $
\theta $, the $ n^{O_c(1)} $ factor is due to the size of the perfect hash
family, the number of choices of $ (\eta'_1,\ldots,\eta'_\theta) $, and the
choice of the first seed, and an additional $ n^{O(1)}\cdot \prod_i d_i $ factor is the number of
expander walks. Simplifying, $|\mathcal{S}_\text{low,1}^{n,c}|=n^{O_c(1)}{\prod
d_i}=n^{O_c(1)}\prod_i(\eta_i')^{-O(1)}\leq n^{O_c(1)}$, where the last
inequality is due to the choice of $ \eta'_i $'s.

\paragraph{Analysis.} We follow the outline. First, by a union bound we know that
$\Pr_{x\sim [m]^n}[T^+_\theta(x)=1]\leq \sum_{|S|=\theta}\prod_{i\in S}p_i$ and hence $\sum_{|S|=\theta}\prod_{i\in S}p_i \geq
\eps$. Second, if we hash the indices $[n]$ into
$\theta$ buckets at random and consider one $S$ with $|S| = \theta$, the
probability that the indices in $S$ are `uniformly spread' (one into each
bucket) is $1/2^{O(\theta)}$. By Lemma \ref{lemma_perfect_hash}, this property is also true if we pick $h$ from the explicit
perfect hash family $\mc{H}^{n,\theta}_{perf}$. 

Formally, given an $h\in \mc{H}^{n,\theta}_{perf}$, define $\alpha_h = \prod_{i
\in [\theta]}\sum_{j\in B_i}p_j$.  Over a uniform choice of $h$ from the
family $\mc{H}^{n,\theta}_{perf}$, we can conclude that 
\[
\E_h \alpha_h \geq
\sum_{|S|=\theta}\prod_{i\in S}p_i \Pr_h[\text{$h$ is 1-1 on $S$}] \geq
\frac{\eps}{2^{O(\theta)}} \geq \frac{1}{n^{O_c(1)}}.
\] 
Thus there must exist an $h$ that
satisfies $\alpha_h \geq 1/n^{O_c(1)}$.

We fix such an $h$. For a bucket $B_i$, define $\eta_i = \Pr_{x \in
\mc{G}^{m,n}_{2-wise}}[\sum_{j\in B_i}1_{A_j}(x_j) \geq 1]$. Now for a moment, let us
analyze the construction assuming \emph{independently seeded} pairwise independent spaces
in each bucket. Then the success probability, namely the probability that
\emph{every} bucket $B_i$ has a non-zero $ \sum_{j\in B_i}1_{A_j}(x_j) $ is equal
to $\prod_i \eta_i$. The following claim gives a lower bound on this probability.

\begin{claim} \label{cl:low-size} 
For the function $h$ satisfying $\alpha_h \geq 1/n^{O_c(1)}$, we have 
$\prod_{i\in[\theta]} \eta_i \geq 1/n^{O_c(1)}$.
\end{claim}

\begin{proof} 
For a bucket $ B_i $, define $ \mu_i = \sum_{j\in B_i} p_j$. Further, 
call a bucket $ B_i $ as being \emph{good} if $\mu_i \leq 1/2 $, otherwise 
call the bucket \emph{bad}. For the bad buckets, 
\begin{equation} \label{eq:bad-buckets}
\prod_{B_i\ bad} \mu_i \leq 
\prod_{B_i\ bad} e^{\mu_i} = \exp\left(\sum_{B_i\ bad} \mu_i\right) \leq e^\mu
\leq n^{O_c(1)}.
\end{equation}
From the choice of $h$ and the definition of $\alpha_h$ we have
\begin{equation} \label{eq:good-buckets}
\frac{1}{n^{O_c(1)}} \leq \prod_{i \in [\theta]} \mu_i =
\prod_{B_i\ bad} \mu_i \prod_{B_i\ good} \mu_i \leq
n^{O_c(1)} \prod_{B_i\ good} \mu_i \Rightarrow 
\prod_{B_i\ good} \mu_i \geq \frac{1}{n^{O_c(1)}},
\end{equation}
where we have used Equation~\eqref{eq:bad-buckets} for the second inequality. 

Now let's analyze the $\eta_i$'s. For a good bucket $B_i$, by
inclusion-exclusion,
\begin{equation} \label{eq:good-bucket}
\eta_i = \prob{x}{\sum_{j \in B_i} 1_{A_j}(x_j) \geq 1} \geq \sum_{j \in B_i} p_j -
\sum_{j,k \in B_i: j < k} p_jp_k
\geq \mu_i - \frac{\mu_i^2}{2} \geq \frac{\mu_i}{2}.
\end{equation}

For a bad bucket, $\mu_i > 1/2$. But since all $p_i$'s are $\leq 1/2$, it isn't
hard to see that 
there must exist a non empty subset $B_{i}' \subset B_i$ satisfying $1/4 \leq
\mu_{i}':=\sum_{j\in B_i'}p_j \leq 1/2$.
We now can use Equation~\eqref{eq:good-bucket} on the good bucket $B_{i}'$ to
get the bound on the 
bad bucket $B_i$ as follows:
\begin{equation} \label{eq:bad-bucket}
\eta_i \geq \Pr_x \left[\sum_{j\in B_i'} 1_{A_j}(x_j) \geq 1\right] \geq \frac{\mu_{i}'}{2} \geq \frac{1}{8}.
\end{equation}
So finally,
\[
\prod_{i \in [\theta]} \eta_i \geq \prod_{B_i\ bad} \frac{1}{8} \prod_{B_i\
good} \frac{\mu_i}{2}
\geq \frac{1}{2^{O(\theta)}}\frac{1}{n^{O_c(1)}} = \frac{1}{n^{O_c(1)}},
\]
where we have used \eqref{eq:good-bucket} and \eqref{eq:bad-bucket}
for the first inequality and \eqref{eq:good-buckets} for the second inequality.
\end{proof}

If now the seeds for $\mc{G}^{m,n}_{2-wise}$ in each bucket are chosen according
to the expander walk ``corresponding'' to the probability vector  $
(\eta_1,\ldots,\eta_\theta) $, then by Lemma~\ref{lemma_afwz} the success
probability becomes at least $ (1/2^{O(\theta)})\prod_i \eta_i \geq
1/n^{O_c(1)}$, using Claim~\ref{cl:low-size} for the final inequality.

But we are not done yet. We cannot guess the correct probability vector exactly.
Instead, we get a closest guess $ (\eta_1',\ldots,\eta_\theta') $ such that for
all $ i \in [\theta] $, $1/\eta_i'$ is a power of $2$ and $ \eta_i'\geq\eta_i/2 $. Again, by
Lemma~\ref{lemma_afwz} the success probability becomes at least $
(1/2^{O(\theta)})\prod_i \eta_i'\geq (1/2^{O(\theta)})^2 \prod_i \eta_i \geq
1/n^{O_c(1)}$, using Claim~\ref{cl:low-size} for the final inequality. Note that this also tells us that it is sufficient to guess $\eta_i'$ such that  $\prod_i (1/\eta_i') \leq n^{O_c(1)}$.

%% file: lowwt-general.tex
\subsection{The general low-weight case}
 \label{sec:lowwt_general}

The general case (where $p_i$ are arbitrary) is more technical: here we need to
do a ``two level'' hashing. The top level is by dividing into buckets, and in
each bucket we get the desired ``advantage'' using a generalization of hitting
sets for combinatorial rectangles (which itself uses hashing) from \cite{LLSZ}.
The theorem we prove for this case can be stated as follows. 
\begin{theorem}
\label{thm_low_wt}
Fix any $c \ge 1$. For any $m\leq n^c$, there exists an explicit $1/n^c$-HS
$\mathcal{S}_\text{low}^{n,c} \subseteq [m]^n$ of size $n^{O_{c}(1)}$ for 
functions $f \in \cthr(m,n)$ s.t. $w(f)\leq c\log n$.
\end{theorem}

\paragraph{Construction.}
We describe $\mc{S}_{low}^{n,c}$ by demonstrating how to sample a random element
$x$ of this set. The number of possible random choices bounds
$|\mc{S}_{low}^{n,c}|$. We define the sampling process in terms of certain
constants $c_i$ ($i\in [5]$) that depend on $c$ in a way that will become clear
later in the proof. Assuming this, it will be clear that
$|\mc{S}_{low}^{n,c}|=n^{O_c(1)}$.

\textbf{Step 1}: Choose at random $t\in \{0,\ldots,15c\log n\}$. If $t=0$, then we simply
output a random element $x$ of $\mc{S}_{LLSZ}^{m,n,1/n^{c_1}}$ for some constant
$c_1$. The number of choices for $t$ is $O_c(\log n)$ and if $t=0$, the number
of choices for $x$ is $n^{O_c(1)}$. The number of choices for non-zero $t$ are
bounded subsequently.

\textbf{Step 2}: Choose $h\in \mc{H}_{perf}^{n,t}$ uniformly at random. The number 
of choices for $h$ is $n^{O_c(1)}\cdot 2^{O(t)} = n^{O_c(1)}$.

\textbf{Step 3}: Choose at random non-negative integers $\rho_1,\ldots,\rho_t$ and
$a_1,\ldots,a_t$ s.t. $\sum_i \rho_i \leq c_2\log n$ and $\sum_i a_i \leq
c_3\log n$. For any constants $c_2$ and $c_3$, the number of choices for
$\rho_1,\ldots,\rho_t$ and $a_1,\ldots,a_t$ is $n^{O_c(1)}$.

\textbf{Step 4}: Choose a set $V$ s.t. $|V| = n^{O_c(1)}=N$ and identify $V$
with $\mc{S}_{rect}^{n,c_4,\rho_i}$ for some constant $c_4\geq 1$ and each $i\in
[t]$ in some arbitrary way (we assume w.l.o.g. that the sets
$\mc{S}_{rect}^{n,c_4,\rho_i}$ ($i\in [t]$) all have the same size). Fix a
sequence of expander graphs $(G_1,\ldots,G_t)$ with vertex set $V$ where $G_i$
is an $(N,D_i,\lambda_i)$-expander with $\lambda_i \leq 1/(10\cdot 2^{a_i}\cdot
2^{a_{i+1}})$ and $D_i = 2^{O(a_i+a_{i+1})}$ (this is possible by Fact
\ref{fact_expl_exp}). Choose $w\in \mc{W}(G_1,\ldots,G_t)$ uniformly at random.
For each $i\in [t]$, the vertex $v_i(w)\in V$ gives us some $x^{(i)}\in
\mc{S}_{rect}^{n,c_4,\rho_i}$. Finally, we set $x\in [m]^n$ so that
$x|_{h^{-1}(i)} = x^{(i)}|_{h^{-1}(i)}$. The total number of choices in this
step is bounded by $|\mc{W}(G_1,\ldots,G_t)|\leq N\cdot \prod_i D_i \leq
n^{O_c(1)}\cdot 2^{O(\sum_i a_i)} = n^{O_c(1)}$.

Thus, the number of random choices (and hence $|\mc{S}_{low}^{n,c}|$) is at most
$n^{O_c(1)}$. 

\paragraph{Analysis.}
We will now prove Theorem~\ref{thm_low_wt}. The analysis once again follows the
outline of Section~\ref{sec:lowwt_lowsize}.

For brevity, we will denote $\mc{S}_{low}^{n,c}$ by $\mc{S}$. Fix any $A_1,\ldots, A_n \subseteq [m]$ and a threshold
test $f\in \cthr(m,n)$ such that $f(x):= T^+_\theta(\sum_{i\in [n]} 1_{A_i}(x_i))$ for some $\theta\in\naturals$ (we can analyze combinatorial thresholds $f$ that use thresholds of the form $T^{-}_\theta$ in a symmetric way). We assume that $f$ has low-weight and good acceptance probability on uniformly random input: that is, $w(f)\leq c\log n$ and
$\prob{x\in [m]^n}{f(x) =1 }\geq 1/n^c$ . For each $i\in [n]$, let $p_i$ denote $|A_i|/m$ and $q_i$
denote $1-p_i$. We call $A_i$ small if $p_i\leq 1/2$ and large otherwise. Let $S
= \setcond{i}{\text{$A_i$ small}}$ and $L = [n]\setminus S$. Note that
$w(f) = \sum_{i}p_iq_i \geq \sum_{i\in S} p_i/2 + \sum_{i\in L}q_i/2$.

Also, given $x\in [m]^n$, let $Y(x) = \sum_{i\in S} 1_{A_i}(x_i)$ and
$\overline{Z}(x) = \sum_{i\in L}1_{\overline{A_i}}(x_i)$. We have $\sum_i
1_{A_i}(x_i) = Y(x) + (|L|-\overline{Z}(x))$ for any $x$. We would like to show
that $\prob{x\in \mc{S}}{f(x)=1} > 0$. Instead we show the following
stronger statement: $\prob{x\in\mc{S}}{\overline{Z}(x)=0 \wedge Y(x)\geq
\theta-|L|}>0$. To do this, we first need the following simple claim.

\begin{claim}
\label{claim_subevent}
$\prob{x\in [m]^n}{\overline{Z}(x) = 0 \wedge Y(x)\geq \theta-|L|}\geq
1/n^{c_1}$, for $c_1 = O(c)$.
\end{claim}

\begin{proof}
Clearly, we have $\prob{x\in [m]^n}{\overline{Z}(x) = 0 \wedge Y(x)\geq \theta-|L|} = \prob{x\in [m]^n}{\overline{Z}(x) = 0}\cdot \prob{x\in [m]^n}{Y(x)\geq \theta-|L|} $. We lower bound each term separately by $1/n^{O(c)}$.

To bound the first term, note that $\prob{x\in [m]^n}{\overline{Z}(x) = 0} = \prod_{i\in L}(1-q_i) = \exp\{-O(\sum_{i\in L} q_i)\}$ where the last inequality follows from the fact that $q_i < 1/2$ for each $i\in L$ and $(1-x)\geq e^{-2x}$ for $x\in [0,1/2]$. Now, since each $q_i < 1/2$, we have $w_i\leq 2q_i$ for each $i\in L$ and hence, $\sum_{i\in L}q_i = O(w(f)) = O(c\log n)$. The lower bound on the first term follows.

To bound the second term, we note that $\prob{x\in [m]^n}{Y(x)\geq \theta'}$ can only decrease as $\theta'$ increases. Thus, we have
\begin{align*}
\prob{x\in [m]^n}{Y(x)\geq \theta - |L|} &= \sum_{i\geq 0}\prob{x\in [m]^n}{Y(x) \geq \theta- |L|}\cdot \prob{x\in [m]^n}{\overline{Z}(x) = i}\\
&\geq \sum_{i\geq 0}\prob{x\in [m]^n}{Y(x) \geq (\theta- |L|+i)\wedge \overline{Z}(x) = i}\\
&= \prob{x\in [m]^n}{\sum_{i\in [n]}1_{A_i}(x_i)\geq \theta}\geq 1/n^c
\end{align*} 

This proves the claim.
\end{proof}

To show that $\prob{x\in\mc{S}}{\overline{Z}(x)=0 \wedge Y(x)\geq
\theta-|L|}>0$, we define a sequence of ``good'' events whose conjunction occurs
with positive probability and which together imply that $\overline{Z}(x)=0$ and
$Y(x)\geq \theta-|L|$. 

Event $\mc{E}_1$: $t = \max\{\theta - |L|,0\}$. Since $f(x) = T^+_\theta(\sum_i 1_{A_i}(x_i))$ accepts a
uniformly random $x$ with probability at least $1/n^c$, we have by Chernoff bounds, we
must have $\theta - \avg{x}{\sum_i 1_{A_i}(x_i)}\leq 10c\log n$. Since
$\avg{x}{\sum_i 1_{A_i}(x_i)} \leq \sum_{i\in S}p_i + \sum_{i\in L}p_i \leq
2w(f) + |L|$, we see that $\theta - |L|\leq 12c\log n$ and hence,
there is some choice of $t$ in Step $1$ so that $\mc{E}_1$ occurs. We condition
on this choice of $t$. Note that by Claim \ref{claim_subevent}, we have
$\prob{x\in [m]^n}{\overline{Z}(x) = 0 \wedge Y(x)\geq t}\geq 1/n^{c_1}$. If
$t=0$, then the condition that $Y(x)\geq t$ is trivial and hence the above event
reduces to $\overline{Z}(x)=0$, which is just a combinatorial rectangle and
hence, there is an $x\in \mc{S}_{LLSZ}^{m,n,1/n^{c_1}}$ with $f(x)=1$
and we are done. Therefore, for the rest of the proof we assume that $t\geq 1$. 

Event $\mc{E}_2$: Given $h\in \mc{H}_{perf}^{n,t}$, define $\alpha_h$ to be the
quantity $\prod_{i\in [t]}\left(\sum_{j\in h^{-1}(i)\cap S}p_j \right)$. Note
that by Lemma \ref{lemma_perfect_hash}, for large enough constant $c_1'$ depending on $c$, we have
\begin{align*}
\avg{h\in \mc{H}_{perf}^{n,t}}{\alpha_h} &\geq \sum_{T\subseteq S:
|T|=t}\prod_{j\in T} p_j \prob{h}{\text{$h$ is $1$-$1$ on $T$}}\\
&\geq \frac{1}{2^{O(t)}}\sum_{T\subseteq S: |T|=t}\prod_{j\in T} p_j\\
&\geq \frac{1}{2^{O(t)}}\prob{x}{Y(x) \geq t}\qquad (\text{by union
bound})\\
&\geq \frac{1}{n^{c_1'}}
\end{align*}

Event $\mc{E}_2$ is simply that $\alpha_h \geq 1/n^{c_1'}$. By averaging, there
is such a choice of $h$. Fix such a choice.

Event $\mc{E}_3$: We say that this event occurs if for each $i\in [t]$, we have
$\rho_i = \lceil \sum_{j\in h^{-1}(i)\cap S} p_j + \sum_{k\in h^{-1}(i)\cap S}
q_k\rceil + 1$. To see that this event can occur, we only need to verify that
for this choice of $\rho_i$, we have $\sum_i\rho_i \leq c_2\log n$ for a
suitable constant $c_2$ depending on $c$. But this straightaway follows from the
fact that $\sum_{j\in S} p_j + \sum_{k\in L} q_k \leq 2w(f)\leq 2c\log
n$. Fix this choice of $\rho_i$ ($i\in [t]$).

To show that there is an $x\in \mc{S}$ s.t. $\overline{Z}(x) = 0$ and $Y(x) \geq
t$, our aim is to show that there is an $x\in \mc{S}$ with $\overline{Z}_i(x) :=
\sum_{j\in h^{-1}(i)\cap L} 1_{\overline{A_j}}(x_j) = 0$ and $Y_i(x) := \sum_{j\in
h^{-1}(i)\cap S} 1_{A_j}(x_j) \geq 1$ for each $i\in [t]$. To show that this
occurs, we first need the following claim.

\begin{claim}
\label{claim_disj_rect}
Fix $i\in [t]$. Let $p_i' = \prob{x\in
\mc{S}_{rect}^{n,c_4,\rho_i}}{\overline{Z}_i(x)=0 \wedge Y_i(x)\geq 1}$. Then,
$p_i'\geq (\sum_{j\in h^{-1}(i)\cap S}p_j)/2^{c_4'\rho_i}$, for large enough
constants $c_4$ and $c_4'$ depending on $c$.
\end{claim}

\begin{proof}[Proof of Claim \ref{claim_disj_rect}]
We assume that $p_j > 0$ for every $j\in h^{-1}(i)\cap S$ (the other $j$ do not
contribute anything to the right hand side of the inequality above).
  
The claim follows from the fact that the event $\overline{Z}_i(x)=0 \wedge
Y_i(x)\geq 1$ is implied by any of the \emph{pairwise disjoint} rectangles
$R_j(x) = 1_{A_j}(x_j)\wedge \bigwedge_{j\neq k\in h^{-1}(i)\cap
S}1_{\overline{A_k}}(x_k) \wedge \bigwedge_{\ell\in h^{-1}(i)\cap
L}1_{A_\ell}(x_\ell)$ for $j\in h^{-1}(i)\cap S$. Thus, we have
\begin{equation}
\label{eq_disj_ineq}
p_i' = \prob{x\in \mc{S}_{rect}^{n,c_4,\rho_i}}{\overline{Z}_i(x)=0 \wedge
Y_i(x)\geq 1}\geq \sum_{j\in h^{-1}(i)\cap S} \prob{x\in
\mc{S}_{rect}^{n,c_4,\rho_i}}{R_j(x)=1}
\end{equation}
However, by our choice of $\rho_i$, we know that $\rho_i \geq \sum_{j\in
h^{-1}(i)\cap S}p_j+ \sum_{k\in h^{-1}(i)\cap S}q_k + 1$, which is at least the
sum of the rejecting probabilities of each combinatorial rectangle $R_j$ above. Moreover, $\rho_i \leq \sum_{s\in [t]}\rho_t \leq c_2\log n$. Below, we choose $c_4\geq c_2$ and so we have $\rho_i \leq c_4\log n$.

Note also that for each $j\in h^{-1}(i)\cap S$, we have 
\begin{align*}
P_j&:=\prob{x\in [m]^n}{R_j(x)=1}\\
&\geq
p_j\prod_{k\in h^{-1}(i)\cap S}(1-p_k)\prod_{\ell\in h^{-1}(i)\cap
L}(1-q_\ell)\\
&\geq p_j \exp\{-2(\sum_{k}p_k + \sum_{\ell}q_\ell)\}
\geq p_j \exp\{-2\rho_i\}
\end{align*}

where the second inequality
follows from the fact that $(1-x)\geq e^{-2x}$ for any $x\in [0,1/2]$. In particular, for large enough constant $c_4 > c_2$, we see that $P_j\geq 1/m\cdot 1/n^{O(c)}\geq 1/n^{c_4}$.

Thus, by Theorem \ref{thm:comb-rect}, we have for each
$j$, $\prob{x\in \mc{S}_{rect}^{n,c_4,\rho_i}}{R_j(x)=1} \geq
P_j/2^{O_c(\rho_i)}$; since $P_j\geq p_j/2^{O(\rho_i)}$, we have $\prob{x\in \mc{S}_{rect}^{n,c_4,\rho_i}}{R_j(x)=1} \geq
p_j/2^{(O_c(1)+O(1))\rho_i} \geq p_j/2^{c_4'\rho_i}$ for a large enough constant $c_4'$ depending on $c$. This bound, together with (\ref{eq_disj_ineq}), proves the claim.
\end{proof}
 
The above claim immediately shows that if we plug in \emph{independent} $x^{(i)}$ chosen at random from
$\mc{S}_{rect}^{n,c_4,\rho_i}$ in the indices in $h^{-1}(i)$, then the
probability that we pick an $x$ such that $\overline{Z}(x) = 0$ and $Y(x) \geq
t$ is at least 
\begin{align}
\prod_i p_i' &\geq 1/2^{O_c(\sum_{i\in [t]}\rho_i)}\prod_{i\in [t]} (\sum_{j\in
h^{-1}(i)\cap S}p_j)\notag\\
&= 1/2^{O_c(\log n)}\cdot \alpha_h\geq 1/n^{O_c(1)}\label{eq_prod_bound}
\end{align} 

However, the $x^{(i)}$ we actually choose
are not independent but picked according to a random walk $w\in
\mc{W}(G_1,\ldots, G_t)$. But by Lemma \ref{lemma_afwz}, we see that for this
event to occur with positive probability, it suffices to have $\lambda_i \leq
p_{i-1}'p_i'/10$ for each $i\in [t]$. To satisfy this, it suffices to have
$1/2^{a_i} \leq p_i' \leq 1/2^{a_i-1}$ for each $i$. This is exactly the
definition of the event $\mc{E}_4$.

Event $\mc{E}_4$: For each $i\in [t]$, we have $1/2^{a_i}\leq p_i' \leq
1/2^{a_i-1}$. For this to occur with positive probability, we only need to check
that $\sum_{i\in [t]} \lceil\log(1/p_i')\rceil \leq c_3\log n$ for large enough constant
$c_3$. But from (\ref{eq_prod_bound}), we have
\begin{align*}
\sum_i \lceil\log(1/p_i')\rceil &\leq \left(\sum_i \log(1/p_i')\right) + t\\
&\leq O_c(\log n) + O(c\log n) \leq c_3\log n
\end{align*}
 for large enough constant $c_3$ depending on $c$. This shows that $\mc{E}_4$ occurs with positive probability and concludes the analysis.

%% file: stronger-rect.tex
As mentioned in the introduction, \cite{LLSZ} give $\eps$-hitting set
constructions for combinatorial rectangles, even for $\eps =
1/\poly(n)$. However in our applications, we require something slightly
stronger -- in particular, we need a set $\calS$ s.t. $\Pr_{x \sim
  \calS} (x$ in the rectangle$) \ge \eps$ (roughly speaking). We
however need to fool only special kinds of rectangles, given by the
two conditions in the following theorem.

\begin{theorem} [Theorem~\ref{thm:comb-rect} restated]
        For all constants $c>0$, $m=n^c$, and $\rho \le c\log n$, 
        for any $\calR \in \crect(m, n)$ which satisfies the properties:
        \begin{enumerate}
            \item $\calR$ is defined by $A_i$, and the rejecting probabilities
            $q_i$ satisfy $\sum_i q_i \le \rho$ and
            \item $p:=\Pr_{x\sim [m]^n} [ \calR(x) = 1] \ge 1/n^c$,
        \end{enumerate}
       there is an
       explicit set $\calS_{\text{rect}}^{n,c,\rho}$ of size $n^{O_c(1)}$ that satisfies $ \Pr_{x\sim \calS_{\text{rect}}^{n,c,\rho}} [ \calR(x) = 1] \ge p/2^{O_c(\rho)}$.
    \end{theorem}

To outline the construction, we keep in mind a rectangle $\calR$
(though we will not use it, of course) defined by sets $A_i$, and
write $p_i = |A_i|/m$, $q_i = 1-p_i$. W.l.o.g., we assume that $\rho\geq 10$. The outline of the construction is as follows:
    \begin{enumerate}
        \item We guess an integer $r \le \rho/10$ (supposed to be an estimate for
        $\sum_i q_i/10$).
        \item Then we use a fractional hash family $\calH^{n, r}_{frac}$ to map the
        indices into $r$ buckets. This ensures that each bucket has roughly
        a constant weight.
        \item In each bucket, we show that taking $O(1)$-wise independent
        spaces (Fact~\ref{fact_k_wise_existence}) ensures a success
        probability (i.e. the probability of being inside $\mc{R}$) depending on the weight of the bucket.
        \item We then combine the distributions for different buckets using
        expander walks (this step has to be done with more care now, since
        the probabilities are different across buckets).
\end{enumerate}
%
%\paragraph{Construction.} We formally describe the construction by describing how to sample a random $x\in \mc{S}_{rect}^{n,c,\rho}$. Let $c_1\in\naturals$ be some constant parameter whose value will be fixed later on in the proof.
%
%\textbf{Step 1:} Choose at random a positive integer $r \leq \rho/10$. (supposed to be an estimate for $\sum_i q_i/10$). 
%
%\textbf{Step 2:} Choose a uniformly random $h:[n]\rightarrow [r]$ from the fractional perfect hash family $\calH^{n, r}_{frac}$ (Lemma \ref{lem:fract-hash}) to map the indices into $r$ buckets. This ensures that each bucket has roughly a constant weight. Note that the number of possibilities for $h$ is $|\mc{H}^{n,r}_{frac}| = 2^{O(r)}n^{O(1)} = n^{O_c(1)}$.
%
%\textbf{Step 3:} Fix a 

Steps (1) and (2) are simple: we try all choices of $r$,
and the `right' one for the hashing in step (2) to work is $r = \sum_i q_i/10$; the probability that we make this correct guess is at least $1/\rho \gg 1/2^{\rho}$.
In this case, by the fractional hashing lemma, we obtain a hash family
$\calH_{frac}^{n,r}$, which has the property that for an $h$ drawn from it,
we have 
\[ \Pr \left[\sum_{j \in h^{-1}(i)} q_j \in [1/100, 100] \text{ for all }i \right] \ge \frac{1}{2^{O_c(r)}}\ge \frac{1}{2^{O_c(\rho)}}. \]

Step (3) is crucial, and we prove the following:
\begin{claim}\label{claim:rect1}
There is an absolute constant $a\in\naturals$ s.t. the following holds.
Let $A_1, \dots, A_k$ be the accepting sets of a combinatorial
rectangle $\calR$ in $\crect(m,k)$, and let $q_1, \dots, q_k$ be {\em
  rejecting} probabilities as defined earlier, with $\sum_i q_i \leq C$,
for some constant $C\geq 1$. Suppose $\prod_i (1-q_i) = \pi$, for some
$\pi>0$. Let $\calS$ be the support of an $aC$-wise independent
distribution on $[m]^n$ (in the sense of Fact~\ref{fact_k_wise_existence}). Then
\[ \Pr_{x \in \calS} [\calR(x)=1] \ge \frac{\pi}{2}. \] 
\end{claim}
\begin{proof}
We observe that if $\sum_i q_i \le C$, then at most $2C$ of the $q_i$
are $\ge 1/2$. Let $B$ (for `big') denote the set of such indices. Now
consider $\calS$, an $aC$-wise independent distribution
over $[m]^n$. Let us restrict to the vectors in the distribution for
which the coordinates corresponding to $B$ are in the rectangle
$R$. Because the family is $aC$-wise independent, the number of such
vectors is precisely a factor $\prod_{i \in B} (1-q_i)$ of the
support of $\calS$.

Now, even after fixing the values at the locations indexed by $B$, the chosen
vectors still form a $(a-2)C$-wise independent
distribution. Thus by Theorem \ref{thm:EGLNV}, we have
that the distribution $\del$-approximates, i.e., maintains the
probability of any event (in particular the event that we are in the
rectangle $\calR$) to an additive error of $\del = 2^{-\Omega((a-2)C)} <
(1/2) e^{-2\sum_i q_i} < (1/2) \prod_{i \not\in B} (1-q_i)$ for large enough $a$ (In the
last step, we used the fact that if $x < 1/2$, then $(1-x) > e^{-2x}$).
Thus if we restrict to coordinates outside $B$, we have that
the probability that these
indices are `accepting' for $\calR$ is at least $(1/2) \prod_{i \not \in B}
p_i$ (because we have a very good additive approximation).

Combining the two, we get that the overall accepting probability is
$\frac{\pi}{2}$, finishing the proof of the claim.
\end{proof}

Let us now see how the claim fits into the argument. Let $B_1, \dots, B_r$
be the sets of indices of the buckets obtained in Step (2). Claim~\ref{claim:rect1}
now implies that if we pick an $aC$-wise independent family on all the
$n$ positions (call this $\calS$), the probability that we obtain a rectangle on $B_i$ is at
least $(1/2) \prod_{j \in B_i} (1-q_j)$. For convenience, let us write
$P_i = (1/2)\prod_{j \in B_i} (1-q_j)$.  We wish to use an expander walk 
argument as before -- however this time the probabilities $P_i$ of success are
different across the buckets.

The idea is to estimate $P_i$ for each $i$, up to a {\em sufficiently small} error.
Let us define $L = \lceil c\log n\rceil$ (where $p$ is as in the statement of Theorem~\ref{thm:comb-rect}).
Note that $L \ge \log(1/p)$, since $p \ge 1/n^c$.  Now, we estimate $\log(1/P_i)$ by the
smallest integer multiple of $L':=\lfloor L/r\rfloor\geq 10$ which is larger than it: call it $\alpha_i\cdot L'$. Since
$\sum_i \log(1/P_i)$ is at most $L$, we have $\sum_i \alpha_i L' \le 2L$,
or $\sum_i \alpha_i \le 3r$. Since the sum is over $r$ indices, there are at most $2^{O(r)}$
choices for the $\alpha_i$ we need to consider.  Each choice of the $\alpha_i$'s gives an
estimate for $P_i$ (which is also a {\em lower bound} on $P_i$). More formally, set
$\rho_i = e^{-\alpha_i L'}$, so we have $P_i \ge \rho_i$ for all $i$.

Finally, let us construct graphs $G_i$ (for $1\le i\le r$) with the vertex set being
$\calS$ (the $aC$-wise independent family), and $G_i$ having a degree depending on
$\rho_i$ (we do this for each choice of the $\rho_i$'s). By the expander walk
lemma~\ref{lemma_afwz}, we obtain an overall probability
of success of at least $\prod_i P_i / 2^{O(r)}$ for the ``right'' choice of the $\rho_i$'s.
Since our choice is right with probability at least $2^{-O(r)}$, we obtain a
success probability in Steps (3) and (4) of at least $\prod_i P_i/2^{O(r)} \ge p/2^{O(r)}\ge p/2^{O(\rho)}$. In combination with the success probability of $1/2^{O_c(\rho)}$ above for Steps (1) and (2), this gives us the claimed overall success probability.

Finally, we note that the total seed length we have used in the process is $O_c(\log n + \sum_i \log (1/\rho_i))$,
which can be upper bounded by $O_c(\log n+ L) = O_c(\log n)$.

%% file: fract-hash.tex
The first step in all of our constructions has been hashing into a smaller number of buckets. To this
effect, we need an explicit construction of hash families which have several ``good'' properties.  In particular, we will prove the following lemma in this section.
\begin{lemma}[Fractional Perfect Hash Lemma: Lemma~\ref{lem:fract-hash} restated]
For any $n,t\in\naturals$ such that $t\leq n$, there is an explicit family of hash functions $\mc{H}_{frac}^{n,t}\subseteq [t]^{[n]}$ of size $2^{O(t)}n^{O(1)}$ such that for any $z\in [0,1]^n$ such that $\sum_{j\in [n]} z_j \geq 10t$, we have
\[
\prob{h\in \mc{H}_{frac}^{n,t}}{\forall i\in [t], 0.01\frac{\sum_{j\in [n]} z_j}{t}\leq \sum_{j\in 
h^{-1}(i)}z_j \leq 10\frac{\sum_{j\in [n]} z_j}{t}}\geq \frac{1}{2^{O(t)}}
\]
\end{lemma}
\begin{proof}
For $S\subseteq [n]$, we define $z(S)$ to be $\sum_{j\in S}z_j$. By assumption, we have 
$z([n])\geq 10t$. Without loss of generality, we assume that $z([n])=10t$ (otherwise, we work with 
$\tilde{z} = (10t/z([n]))z$ which satisfies this property; since we prove the lemma for 
$\tilde{z}$, it is true for $z$ as well). We thus need to construct $\mc{H}_{frac}^{n,t}$ such that
\[
\prob{h\in \mc{H}_{frac}^{n,t}}{\forall i\in [t], z(h^{-1}(i)) \in [0.1,100]}\geq \frac{1}{2^{O(t)}}
\]
for some constant $c_{frac}>0$.

We describe the formal construction by describing how to sample a random element $h$ of $\mc{H}^{n,t}_{frac}$. To sample a random $h\in \mc{H}^{n,t}_{frac}$, we do the following:

\textbf{Step 1} (Top-level hashing): We choose a pairwise independent hash function $h_1:[n]\rightarrow [10t]$ by choosing a random seed to generator $\mc{G}^{t,n}_{2-wise}$. By Fact \ref{fact_k_wise_existence}, this requires $O(\log n + \log t) = O(\log n)$ bits.

\textbf{Step 2} (Guessing bucket sizes): We choose at random a subset $I'\subseteq [10t]$ of size exactly $t$ and $y_1,\ldots,y_{10t}\in\naturals$ so that $\sum_i y_i \leq 30t$. It can be checked that the number of possibilities for $I'$ and $y_1,\ldots,y_{10t}$ is only $2^{O(t)}$.

\textbf{Step 3} (Second-level hashing): By Fact \ref{fact_k_wise_existence}, for each $i\in [10t]$, we have an explicit pairwise independent family of hash functions mapping $[n]$ to $[y_i]$ given by $\mc{G}^{y_i,n}_{2-wise}$. We assume w.l.o.g. that each such generator has some fixed seedlength $s = O(\log n)$ (if not, increase the seedlength of each to the maximum seedlength among them). Let $V = \{0,1\}^s$. Using Fact \ref{fact_expl_exp}, fix a sequence $(G_1,\ldots,G_{10t})$ of $10t$ many $(2^s, D, \lambda)$-expanders on set $V$ with $D=O(1)$ and $\lambda \leq 1/100$. Choosing $w\in \mc{W}(G_1,\ldots,G_{10t})$ uniformly at random, set $h_{2,i}:[n]\rightarrow [y_i]$ to be $\mc{G}^{y_i,n}_{2-wise}(v_i(w))$. Define, $h_2:[n]\rightarrow [30t]$ as follows: 
\[
h_2(j) = \left(\sum_{i < h_1(j), i\not\in I'} y_i\right) + h_{2,h_1(j)}(j)
\]
Given the random choices made in the previous steps, the function $h_2$ is completely determined by $|\mc{W}(G_1,\ldots,G_{10t})|$, which is $2^{O(t)}$.

\textbf{Step 4} (Folding): This step is completely deterministic given the random choices made in 
the previous steps. We fix an arbitrary map $f:(I'\times \{0\}) \cup 
([30t]\times\{1\})\rightarrow 
[t]$ with the following properties: (a) $f$ is $1$-$1$ on $I'\times \{0\}$, (b) $f$ is $30$-to-$1$ 
on $[30t]\times\{1\}$. We now define $h:[n]\rightarrow [t]$. Define $h(j)$ as
\[
h(j) = \left\{
\begin{array}{lr}
f(h_1(j),0) & \text{if $h_1(j)\in I'$,}\\
f(h_{2,h_1(j)}(j),1) & \text{otherwise.}
\end{array}
\right.
\] 
It is easy to check that $|\mc{H}_{frac}^{n,t}|$, which is the number of possibilities for the random choices made in the above steps, is bounded by $2^{O(t)}n^{O(1)}$, exactly as required. 

We now show that a random $h\in \mc{H}_{frac}^{n,t}$ has the properties stated in the lemma. Assume $h$ is sampled as above. We analyze the construction step-by-step. First, we
recall the following easy consequence of the Paley-Zygmund inequality:
\begin{fact}
\label{fact_PZ}
For any non-negative random variable $Z$ we have
\[ \prob{}{Z\geq 0.1~\E[Z]}\geq 0.9 ~\frac{(\E[Z])^2}{\E[Z^2]}.\]
\end{fact}
Consider $h_1$ sampled in the first step. Define, for each $i\in [30t]$, the random variables $X_i 
= z(h_1^{-1}(i))$ and $Y_i = \sum_{j_1\neq j_2: h_1(j_1)=h_1(j_2)=i}z_{j_1}z_{j_2}$, and let $X = 
\sum_{i\in [10t]} X_i^2$ and $Y = \sum_{i\in [10t]} Y_i$. An easy calculation shows that $X = z_i^2 + Y\leq 10t + Y$. Hence, 
$\avg{h_1}{X} \leq 10t + \avg{h_1}{Y}$ and moreover
\begin{align*}
\avg{h_1}{Y} = \sum_{j_1\neq j_2} z_{j_1}z_{j_2}\prob{h_1}{h_1(j_1) = h_1(j_2)} \leq z([n])^2/10t = 10t
\end{align*}
Let $\mc{E}_1$ denote the event that $Y\leq 20t$. By Markov's inequality, this happens with probability at least $1/2$. We condition on any choice of $h_1$ so that $\mc{E}_1$ occurs. Note that in this case, we have $X \leq 10t+Y\leq 30t$.

Let $Z = X_i$ for a randomly chosen $i\in [10t]$. Clearly, we have $\avg{i}{Z}= (1/10t)\sum_i X_i = 1$ and also $\avg{i}{Z^2} = (1/10t)\sum_iX_i^2 = (1/10t)X\leq 3$. Thus, Fact \ref{fact_PZ} implies that for random $i\in [n]$, we have $\prob{i}{Z \geq 0.1}\geq 0.3$. Markov's Inequality tells us that $\prob{i}{Z> 10}\leq 0.1$. Putting things together, we see that if we set $I = \setcond{i\in [n]}{X_i \in [0.1,10]}$, then $|I|\geq 0.2\times 10t = 2t$. We call the $i\in I$ the \emph{medium-sized buckets}. 

We now analyze the second step. We say that event $\mc{E}_2$ holds if (a) $I'$ contains \emph{only} medium-sized buckets, and (b) for each $i\in [t]$, $y_i = \lceil Y_i \rceil$. Since the number of random choices in Step $2$ is only $2^{O(t)}$ and there are more than $t$ many medium-sized buckets, it is clear that $\prob{}{\mc{E}_2}\geq 1/2^{O(t)}$. We now condition on random choices in Step $2$ so that both $\mc{E}_1$ and $\mc{E}_2$ occur.

For the third step, given $i\not\in I'$, we say that hash function $h_{2,i}$ is \emph{collision-free} if 
for each $k\in [y_i]$, we have $z(S_{i,k})\leq 2$ where $S_{i,k} = h_{2,i}^{-1}(k)\cap 
h_1^{-1}(i)$. The following simple claim shows that this condition is implied by the condition 
that for each $k$, $Y_{i,k} := \sum_{j_1\neq j_2\in S_{i,k}}z_{j_1}z_{j_2} \leq 2$.

\begin{claim}
\label{claim_union}
For any $\alpha_1,\ldots,\alpha_m\in [0,1]$, if $\sum_j \alpha_j > 2$, then $\sum_{j_1\neq j_2}\alpha_{j_1}\alpha_{j_2}> 2$.
\end{claim}

For the sake of analysis, assume first that the hash functions $h_{2,i}$ ($i\in [10t]$) are chosen to be pairwise independent and \emph{independent of each other}. Now fix any $i\in [10t]$ and $k\in [y_i]$. Then, since $h_{2,i}$ is chosen to be pairwise-independent, we have 
\begin{align*}
\avg{}{Y_{i,k}} = \sum_{j_1\neq j_2: h_{2,i}(j_1) = h_{2,i}(j_2) = i} z_{j_1}z_{j_2}\prob{h_{2,i}}{h_{2,}i(j_1) = h_{2,i}(j_2)=k} = Y_i/y_i^2 \leq 1/y_i
\end{align*}
In particular, by Markov's inequality, $\prob{}{Y_{i,k}\geq 2}\leq 1/2y_i$.
Thus, by a union bound over $k$, we see that the probability that a uniformly random pairwise independent hash function $h_{2,i}$ is collision-free is at least $1/2$. 

Now, let us consider the hash functions $h_{2,i}$ as defined in the above construction. Let $\mc{E}_3$ denote the event that for each $i\not\in I'$, $h_{2,i}$ is collision-free. Hence, by Lemma \ref{lemma_afwz}, we see that
\[
\prob{}{\mc{E}_3}=\prob{w\in \mc{W}(G_1,\ldots,G_{10t})}{\forall i\in [10t]\setminus I':\ \text{$h_{2,i}$ collision-free}}\geq 1/2^{O(t)} 
\]

Thus, we have established that $\prob{}{\mc{E}_1\wedge \mc{E}_2\wedge \mc{E}_3}\geq 1/2^{O(t)}$. We now see that when these events occur, then the sampled $h$ satisfies the properties we need. Fix such an $h$ and consider $i\in [t]$. 

Since $f$ is a bijection on $I'\times \{0\}$, we see that there must be an $i'\in I'$ s.t. $f(i') = i$. Since $i'\in I'$ and the event $\mc{E}_2$ occurs, it follows that $i'$ is a medium-sized bucket. Thus,  $z(h^{-1}(i))\geq z(h_1^{-1}(i'))\geq 0.1$. Secondly, since $\mc{E}_3$ occurs, we have
\[
z(h^{-1}(i))=z(h_1^{-1}(i'))+\sum_{\ell\in f^{-1}(i)}z(h_2^{-1}(\ell)\setminus h_1^{-1}(I'))\leq 
10 + 30 \max_{i\in [10t], k\in[y_i]} z(S_{i,k}) \leq 100
\]
where the final inequality follows because $\mc{E}_3$ holds. This shows that for each $i$, we have 
$z(h^{-1}(i))\in [0.1,100]$ and hence $h$ satisfies the required properties. This concludes the 
proof of the lemma.
\end{proof}

%% file: appendix.tex
\section{Proof of the Expander Walk Lemma}

In this section we prove Lemma~\ref{lemma_afwz}. For convenience we restate it below.

\begin{lemma}[Lemma~\ref{lemma_afwz} restated]
Let $G_1,\ldots,G_\ell$ be a sequence of graphs defined on the same vertex set $V$ of size $N$. Assume that $G_i$ is an $(N,D_i,\lambda_i)$-expander. Let $V_1,\ldots,V_\ell\subseteq V$ s.t. $|V_i|\geq p_i N > 0$ for each $i\in [\ell]$. Then, as long as for each $i \in [\ell]$, $\lambda_i \leq (p_{i} p_{i-1}) / 8$,
\begin{equation}\label{eqn:afwz}
\prob{w\in \mc{W}(G_1,\ldots,G_\ell)}{\forall i\in [\ell], v_i(w)\in V_i} \geq (0.75)^\ell \prod_{i\in [\ell]}{p_i}.
\end{equation}
\end{lemma}

Without loss of generality, we can assume that each subset $V_i$ ($i \in [\ell]$) has size \emph{exactly} $p_i N$. 

Let us consider an $\ell$ step random walk starting at a uniformly random starting vertex in $[N]$, in which step $i$ is taken in the graph $G_i$. The probability distribution after $\ell$ steps is now given by $A_1 A_2 \dots A_\ell \ones_N$, where $\ones_N$ denotes the vector $(1/N, \dots, 1/N)$, and $A_i$ is the normalized adjacency matrix of the graph $G_i$.

Now, we are interested in the probability that a walk satisfies the property that its $i$th vertex is in set $V_i$ for each $i$. For $\ell=1$, for example, this is precisely the $L_1$ weight of the set $V_1$, in the vector $A_1 \ones_N$. More generally, suppose we define the operator $I_S$ to be one which takes a vector and returns the ``restriction'' to $S$ (and puts zero everywhere else), we can write the probability as $\norm{I_{V_1} A_1 \ones_N}_1$. In general, it is easy to see that we can write the probability that the $i$th vertex in the walk is in $V_i$ for all $1\le i\le t$ is precisely $\norm{I_{V_t} A_t I_{V_{t-1}} A_{t-1} \dots I_{V_1} A_1}_1$. We will call the vector of interest $u_{(t)}$, for convenience, and bound $\norm{u_{(t)}}_1$ inductively.

Intuitively, the idea will be to show that $u_{(t)}$ should be a vector with a `reasonable mass', and is distributed `roughly uniformly' on the set $V_t$. Formally, we will show the following inductive statement. Define $u_{(0)} = \ones_N$.

\begin{lemma}\label{lem:walk-induct}
For all $1 \le t \le \ell$, we have the following two conditions
\begin{align}
\norm{u_{(t)}}_1 \ge \frac{3p_t}{4} \norm{u_{(t-1)}}_1 \\
\norm{u_{(t)}}_2 \le \frac{2}{\sqrt{p_t N}} \norm{u_{(t)}}_1
\end{align}
\end{lemma}
Note that the second equation informally says that the mass of $u_{(t)}$ is distributed roughly equally on a set of size $p_t N$. The proof is by induction, but we will need a bit of simple notation before we start. Let us define $\parr{u}$ and $\prp{u}$ to be the components of a vector $u$ which are parallel and perpendicular (respectively) to the vector $\ones_N$. Thus we have $u = \parr{u} + \prp{u}$ for all $u$. The following lemma is easy to see.

\begin{claim}\label{lem:walk-trivial}
For any $N$-dimensional vector $x$ with all positive entries, we have $\norm{\parr{x}}_1 = \norm{x}_1$. Furthermore, $\parr{x}$ is an $N$-dimensional vector with each entry $\norm{x}_1/N$.
\end{claim}
\begin{proof}
The ``furthermore'' part is by the definition of $\parr{x}$, and the first part follows directly from it.
\end{proof}

We can now prove Lemma~\ref{lem:walk-induct}. We will use the fact that $A_i \ones_N = \ones_N$ for each $i$, and that $\norm{A_i u}_2 \le \lambda \norm{u}_2$ for $u$ orthogonal to $\ones_N$.

\begin{proof}[Proof of Lemma~\ref{lem:walk-induct}.] For $t=1$, we have $u_{(1)} = I_{V_1} A_1 \ones_N = I_{V_1} \ones_N$, and thus we have $\norm{u_{(1)}}_1 = p_1$, and we have $\norm{u_{(1)}}_2 = \frac{p_1}{\sqrt{p_1 N}}$, and thus the claims are true for $t=1$. Now suppose $t \ge 2$, and that they are true for $t-1$.

For the first part, we observe that
\begin{equation}\label{eqn:walk-induct1}
\norm{u_{(t)}}_1 = \norm{I_{V_t} A_t u_{(t-1)}}_1 \ge \norm{I_{V_t} A_t \parr{u_{(t-1)}}}_1 - \norm{I_{V_t} A_t \prp{u_{(t-1)}}}_1
\end{equation}
The first term is equal to $\norm{I_{V_t} \parr{u_{(t-1)}}}_1 = p_t \norm{u_{(t-1)}}_1$, because $I_{V_t}$ preserves $p_t N$ indices, and each has a contribution of $\norm{u_{(t-1)}}_1/N$, by Claim~\ref{lem:walk-trivial}.

The second term can be upper bounded as
\[ \norm{I_{V_t} A_t \prp{u_{(t-1)}}}_1 \le \sqrt{N} \norm{I_{V_t} A_t \prp{u_{(t-1)}}}_2 \le \sqrt{N} \cdot \lambda_t \norm{u_{(t-1)}}_2 \le \frac{2\lambda_t \sqrt{N}}{\sqrt{p_{t-1} N}} \norm{u_{(t-1)}}_1, \]
where we used the inductive hypothesis in the last step. From the condition $\lambda_t \le p_t p_{t-1}/8$, we have that the term above is bounded above by $p_t \norm{u_{(t-1)}}_1/4$. Combining this with Eq.\eqref{eqn:walk-induct1}, the first inequality follows.

The second inequality is proved similarly. Note that for this part we can even assume the first inequality for $t$, i.e., $\norm{u_{(t)}}_1 \ge (3/4) p_t \norm{u_{(t-1)}}_1$. We will call this (*).
\begin{equation}\label{eqn:walk-induct2}
\norm{u_{(t)}}_2 \le \norm{I_{V_t} A_t \parr{u_{(t-1)}}}_2 + \norm{I_{V_t} A_t \prp{u_{(t-1)}}}_2
\end{equation}
The first term is the $\ell_2$ norm of a vector with support $V_t$, and each entry $\norm{u_{(t-1)}}_1/N$, from Claim~\ref{lem:walk-trivial} we have that the first term is equal to $\frac{\norm{u_{(t-1)}}_1}{N} \cdot \sqrt{p_t N} \le (4/3) \cdot \frac{\norm{u_{(t)}}_1}{\sqrt{p_t N}}$, with the inequality following from (*).

The second term can be bounded by
\[ \lambda_t \norm{u_{(t-1)}}_2 \le \frac{2\lambda_t}{\sqrt{p_{t-1} N}} \cdot \norm{u_{(t-1)}}_1 \le \frac{1}{4\sqrt{p_t N}} \norm{u_{(t)}}_1. \]
Here we first used the inductive hypothesis, and then used (*), along with our choice of $\lambda_t$. Plugging these into Eq.~\eqref{eqn:walk-induct2}, we obtain the second inequality.

This completes the inductive proof of the two inequalities.
\end{proof}

%% file: hitting.bbl
\begin{thebibliography}{10}

\bibitem{AleliunasKaLiLoRa79}
Romas Aleliunas, Richard~M. Karp, Richard~J. Lipton, L{\'a}szl{\'o} Lov{\'a}sz,
  and Charles Rackoff.
\newblock Random walks, universal traversal sequences, and the complexity of
  maze problems.
\newblock In {\em 20th Annual Symposium on Foundations of Computer Science},
  pages 218--223, San Juan, Puerto Rico, 29--31 October 1979. IEEE.

\bibitem{AFWZ}
Noga Alon, Uriel Feige, Avi Wigderson, and David Zuckerman.
\newblock Derandomized graph products.
\newblock {\em Computational Complexity}, 5:60--75, 1995.
\newblock 10.1007/BF01277956.

\bibitem{AYZ}
Noga Alon, Raphael Yuster, and Uri Zwick.
\newblock Color-coding.
\newblock {\em J. ACM}, 42(4):844--856, 1995.

\bibitem{ArmoniSaWiZh96}
Roy Armoni, Michael Saks, Avi Wigderson, and Shiyu Zhou.
\newblock Discrepancy sets and pseudorandom generators for combinatorial
  rectangles.
\newblock In {\em 37th Annual Symposium on Foundations of Computer Science
  (Burlington, VT, 1996)}, pages 412--421. IEEE Comput. Soc. Press, Los
  Alamitos, CA, 1996.

\bibitem{BlumKalaiWasserman}
Avrim Blum, Adam Kalai, and Hal Wasserman.
\newblock Noise-tolerant learning, the parity problem, and the statistical
  query model.
\newblock {\em J. ACM}, 50(4):506--519, 2003.

\bibitem{EGLNV}
Guy Even, Oded Goldreich, Michael Luby, Noam Nisan, and Boban
  Veli{\u{c}}kovi{\'c}.
\newblock Efficient approximation of product distributions.
\newblock {\em Random Structures Algorithms}, 13(1):1--16, 1998.

\bibitem{probability-text}
William Feller.
\newblock {\em An Introduction to Probability Theory and its Applications, Vol
  2}.
\newblock Wiley, 1971.

\bibitem{FKS}
Michael~L. Fredman, J{\'a}nos Koml{\'o}s, and Endre Szemer{\'e}di.
\newblock Storing a sparse table with 0(1) worst case access time.
\newblock {\em J. ACM}, 31(3):538--544, 1984.

\bibitem{GMRZ}
Parikshit Gopalan, Raghu Meka, Omer Reingold, and David Zuckerman.
\newblock Pseudorandom generators for combinatorial shapes.
\newblock In {\em STOC}, pages 253--262, 2011.

\bibitem{HLW}
Shlomo Hoory, Nathan Linial, and Avi Wigderson.
\newblock Expander graphs and their applications.
\newblock {\em Bulletin of the AMS}, 43(4):439--561, 2006.

\bibitem{IW97}
Russell Impagliazzo and Avi Wigderson.
\newblock {$\mathit{P} = \mathit{BPP}$} if {$E$} requires exponential circuits:
  Derandomizing the {XOR} lemma.
\newblock In {\em Proceedings of the Twenty-Ninth Annual ACM Symposium on
  Theory of Computing}, pages 220--229, El Paso, Texas, 4--6 May 1997.

\bibitem{LLSZ}
Nathan Linial, Michael Luby, Michael Saks, and David Zuckerman.
\newblock Efficient construction of a small hitting set for combinatorial
  rectangles in high dimension.
\newblock {\em Combinatorica}, 17:215--234, 1997.
\newblock 10.1007/BF01200907.

\bibitem{LRTV}
Shachar Lovett, Omer Reingold, Luca Trevisan, and Salil Vadhan.
\newblock Pseudorandom bit generators fooling modular sums.
\newblock In {\em Proceedings of the 13th International Workshop on
  Randomization and Computation (RANDOM)}, Lecture Notes in Computer Science,
  pages 615--630. Springer-Verlag, 2009.

\bibitem{Lu02}
Chi-Jen Lu.
\newblock Improved pseudorandom generators for combinatorial rectangles.
\newblock {\em Combinatorica}, 22(3):417--434, 2002.

\bibitem{MekaZu09}
Raghu Meka and David Zuckerman.
\newblock Small-bias spaces for group products.
\newblock In {\em APPROX-RANDOM}, pages 658--672, 2009.

\bibitem{MoserTardos}
Robin~A. Moser and G{\'a}bor Tardos.
\newblock A constructive proof of the general lov{\'a}sz local lemma.
\newblock {\em J. ACM}, 57(2), 2010.

\bibitem{NN93}
Joseph Naor and Moni Naor.
\newblock Small-bias probability spaces: Efficient constructions and
  applications.
\newblock {\em SIAM Journal on Computing}, 22(4):838--856, August 1993.

\bibitem{Nisan92}
Noam Nisan.
\newblock Pseudorandom generators for space-bounded computation.
\newblock {\em Combinatorica}, 12(4):449--461, 1992.

\bibitem{NW}
Noam Nisan and Avi Wigderson.
\newblock Hardness vs. randomness.
\newblock {\em J. Comput. Syst. Sci.}, 49(2):149--167, 1994.

\bibitem{NZ}
Noam Nisan and David Zuckerman.
\newblock Randomness is linear in space.
\newblock {\em Journal of Computer and System Sciences}, 52(1):43--52, February
  1996.

\bibitem{RabaniShpilka}
Yuval Rabani and Amir Shpilka.
\newblock Explicit construction of a small epsilon-net for linear threshold
  functions.
\newblock {\em SIAM J. Comput.}, 39(8):3501--3520, 2010.

\bibitem{SchmidtSiegel}
Jeanette~P. Schmidt and Alan Siegel.
\newblock The analysis of closed hashing under limited randomness (extended
  abstract).
\newblock In {\em STOC}, pages 224--234, 1990.

\bibitem{SU06}
Ronen Shaltiel and Christopher Umans.
\newblock Pseudorandomness for approximate counting and sampling.
\newblock {\em Computational Complexity}, 15(4):298--341, 2006.

\end{thebibliography}
